\documentclass[12pt]{amsart}
\usepackage{amsthm}
\usepackage{amsmath}
\usepackage{amssymb}
\usepackage{amscd}

\newtheorem{theorem}[subsection]{Theorem}
\newtheorem{lemma}[subsection]{Lemma}
\newtheorem{proposition}[subsection]{Proposition}
\newtheorem{corollary}[subsection]{Corollary}

\theoremstyle{definition}

\theoremstyle{remark}

\newtheorem{remark}[subsection]{Remark}

\newtheorem{example}[subsection]{Example}

\newcommand{\Id}{\operatorname{Id}}

\newcommand{\SL}{\operatorname{SL}}
\newcommand{\GL}{\operatorname{GL}}

\newcommand{\gl}{\operatorname{gl}}

\vspace{1.5cm}
\title{{\bf Strict versions of various matrix hierarchies\\
related to ${\rm {\bf \SL_{n}}}$-loops and their combinations}}

\author{G.F. Helminck,\\
Korteweg-de Vries Institute\\
University of Amsterdam\\
P.O. Box 94248\\
1090 GE Amsterdam\\
The Netherlands\\
e-mail: g.f.helminck@uva.nl}


\begin{document}
\maketitle

\begin{abstract}
Let $\mathfrak{t}$ be a commutative Lie subalgebra of ${\rm sl}_{n}(\mathbb{C})$ of maximal dimension. We consider in this paper three spaces of $\mathfrak{t}$-loops that each get deformed in a different way. We require that the deformed generators of each of them evolve w.r.t. the commuting flows they generate according to a certain, different set of Lax equations. This leads to three integrable hierarchies: the $({\rm sl}_{n}(\mathbb{C}), \mathfrak{t})$-hierarchy, its strict version and the combined $({\rm sl}_{n}(\mathbb{C}), \mathfrak{t})$-hierarchy. For $n=2$ and $\mathfrak{t}$ the diagonal matrices, the $({\rm sl}_{2}(\mathbb{C}), \mathfrak{t})$-hierarchy is the AKNS-hierarchy. We treat their interrelations and show that all three have a zero curvature form. Furthermore, we discuss their linearization and we conclude by giving the construction of a large class of solutions.
\end{abstract}

\section{\boldmath Introduction} 
\label{intro}

A simple but fundamental observation in the development of physics and mathematics today is
that breakthroughs in quantum field theory and string theory are usually characterized by the
presence of integrable hierarchies. We illustrate this with several examples.

Two-dimensional topological field theory was solved, see \cite{Witten91}, by Witten in the early 90's and it inspired him to many conjectures,
relating generating functions of intersection numbers of Morita-Mumford classes, matrix models and classical integrable systems of Khadomtsev-Petviashvilii type, see \cite{Witten2}. For the Korteweg de Vries equation the first step of this conjecture was proved by Kontsevich in \cite{Kontsevich94}, and the second by Kharchev, Marshakov, Mironov, Morozov and Zabrodin, \cite{KMMMZ}. A different approach to these conjectures was initiated by work of 
Givental, who defined, see \cite{Giv2004}, a group action on the space of Gromov-Witten potentials and proved it to be transitive on the space of semi-simple potentials. Moreover, Givental gave in  \cite{Giv2003} a link between $A_{r-1}$ singularities and the $r$-KdV hierarchy, $r \geqslant 2$, and showed that his construction is compatible with Witten's conjecture on the relation between the intersection theory on the space of $r$-spin structures on stable curves and the $r$-KdV hierarchy. Using this and their work on tautological relations, Faber, Shadrin and Zvonkine proved in \cite{F-Shadrin-Z2010} that 
 Witten's $r$-spin conjecture is true.

The second example is formed by Seiberg and Witten's complete
solution of four-dimensional $N=2$ supersymmetric Yang-Mills theory. 
Integrability in Seiberg-Witten theory was discovered by I.Krichever, A.Gorsky, A.Marshakov, A.Mironov and A.Morozov, see \cite{KGMMM},
where it was shown that the effective low energy
partition functions is a $\tau$-function of the Whitham hierarchy related to
Hitchin type integrable systems. Also, the extended version of Seiberg-Witten theory, see \cite{A-N2007}, has a link with an integrable hierarchy.
A detailed treatment of integrability in Seiberg-Witten theory can be found in \cite{Marshakov99}.

The third example is the
AdS/CFT correspondence formulated by Maldacena in 1997, stating that a string theory on
Anti-de Sitter space is equivalent to a conformal field theory on its boundary. It has become clear
, by the work of several authors, see e.g. \cite{A-F-R-T2003}, \cite{A-F2009} and \cite{KMMZ}, that on both sides of this
correspondence integrable systems, like spin chains and sigma models, play an important role.
These integrable systems may be used to check the correspondence.

Dealing with supersymmetric gauge theories, one is naturally led to consider $N=1$ theories.
The extremal values of the superpotential in these theories turn out to be described by a
prepotential of some Seiberg-Witten theory, see \cite{CIV}.
Therefore, the whole machinery of Seiberg-Witten theory is applicable in this theory too. 
Moreover, Dijkgraaf and Vafa
have associated in a series of papers, \cite{DV0206}, \cite{DV0207} and \cite{DV0208}, the corresponding prepotential with the logarithm of the partition function of the
Hermitian one-matrix model in the leading order of the size $N$ of the matrix. This leading order
of the matrix model is described by the Whitham hierarchy.

Another exciting common meeting ground for physicists and mathematicians is that of mirror symmetry. This duality developed in the mid-80's out of the observation that a string propagating on a circle with radius $R$ is physically equivalent to a string propagating on a circle with radius $\frac{1}{R}$. In order to be mathematically consistent, one has to require in string theory that some extra dimensions are to be added to spacetime, but on a different scale. 
Thus one arrives at a so-called compactification of the theory. The form that these additional dimensions should have became clear after the 1985 paper, \cite{CHSW85}, by Candelas, Horowitz, Strominger, and Witten. They showed that by compactifying string theory on a Calabi-Yau manifold, one obtains a theory roughly similar to the standard model of particle physics that also consistently incorporates the idea of supersymmetry. In \cite{D88} and \cite{Lerche-V-W89} it was shown that one could not reconstruct back from the compactification, the Calabi-Yau manifold that was used to model the extra dimensions. In trying to do so, it led to two options: type IIA and type IIB and this pair determines the duality. In 1990 Witten introduced in \cite{Witten90} a simplified version of string theory: topological string theory and mirror symmetry survived in this new theory, see \cite{Witten92} and \cite{V92}.
Mirror symmetry attracted the interest of mathematicians around 1990 when Candelas, de la Ossa, Green, and Parkes
showed in \cite{COGP91} that it could be used as a tool in enumerative geometry.
The first striking relation with an integrable hierarchy came, when Givental proved in \cite{Giv-B95} that Toda hierarchies lay at the foundation of Gromov-Witten invariants of projective space. They are the key element in his proof of mirror symmetry for these spaces, see \cite{Giv96}.
By now it has become clear that topological strings on Calabi-Yau geometries, see \cite{A-Dijkgraaf-K-M-V06}, form a unifying picture, where topics as non-critical (super)strings, mirror symmetry, integrable hierarchies and various matrix models connect.

In \cite{F-Takebe90}, Fukuma and Takebe showed that the Toda lattice hierarchy is relevant for the description of deformations of conformal theories, while the KP hierarchy describes unperturbed conformal theories. 
One of the things, they proved was that the $N$-reduced system of the Toda lattice hierarchy
corresponds to the coset conformal model constructed from the affine Lie algebra $\hat{{\rm sl}}_{n}$ as investigated by Eguchi and Yang in \cite{EY89}.

The above examples illustrate sufficiently that integrable hierarchies play a crucial role in various parts of theoretical physics and that it is important to have a profound insight 
into their structure. In the present paper we work like Fukuma and Takebe with ${\rm sl}_{n}$-loops and we consider specific deformations of three Lie subalgebras of $\mathfrak{t}$-loops, with $\mathfrak{t}$ a commutative complex Lie subalgebra of ${\rm sl}_{n}(\mathbb{C})$. The deformed generators of each space of $\mathfrak{t}$-loops should satisfy a certain set of Lax equations and this defines three integrable hierarchies.
We treat the algebraic and geometric properties of these hierarchies, their interrelations and we construct a wide class of solutions. Since there is in our considerations no need, like in \cite{DS}, \cite{R-ST79} and \cite{R-ST81}, to choose $\mathfrak{t}$ to be a Cartan subalgebra, we drop this condition and obtain thus a wider class of deformations.

Recall that integrable hierarchies often occur as the evolution equations of the generators of a deformation of a commutative Lie subalgebra $\mathfrak{c}$ of some Lie algebra $\mathfrak{g}$. Both the deformation and the evolution equations are determined by a splitting of $\mathfrak{g}$
in the direct sum of two Lie subalgebras, like in the Adler-Kostant-Symes Theorem, see \cite{A-vM-VH}. This gives then rise to a compatible set of Lax equations, a so-called hierarchy and the simplest nontrivial equation often determines the name of the hierarchy.

In our situation we take for $\mathfrak{c}$ three spaces of $\mathfrak{t}$-loops. Note that one can just as well choose the commutative Lie subalgebra $\mathfrak{t}$ of ${\rm sl}_{n}(\mathbb{C})$
to be of maximal dimension to include as much commuting flows as possible and this will be done from now on. 
Two concrete examples one can think of, are the diagonal matrices $\mathfrak{t}_{d}$ in ${\rm sl}_{n}(\mathbb{C})$ or its unipotent counterpart
\begin{equation*}
\mathfrak{t}_{u}=\left\{h=\sum_{i=1}^{n-1} a_{i} B^{i} \text{ with } \text{ all }a_{i} \in \mathbb{C} \text{ and }B=\left(\begin{matrix}
0&1&0&\hdots &0\\
0& \ddots &\ddots & \ddots&\vdots\\
\vdots&\ddots &\ddots & \ddots& 0\\
\vdots&\ddots &\ddots & \ddots& 1\\
0& \hdots &\hdots &0 &0
\end{matrix}\right)
\right\}.
\end{equation*}
Let the $\{ E_{\alpha} \mid 1 \leqslant \alpha  \leqslant r \}$ be a basis of $\mathfrak{t}$ and let $z$ be the loop parameter.
Our first choice for $\mathfrak{c}$ is the space $C_{\geqslant 0}$
of all polynomial loops with values in $\mathfrak{t}$ with the basis $\{ E_{\alpha}z^{i} \mid i \geqslant 0, 1 \leqslant \alpha  \leqslant r \}.$ We are interested in certain deformations of the $\{ E_{\alpha} \}$ where the evolution equations of these perturbed loops are determined by the projection of ${\rm sl_{n}}$-loops onto their part containing only positive powers of $z$. This leads to the $({\rm sl}_{n}(\mathbb{C}), \mathfrak{t})$-hierarchy. 

\begin{example}
\label{E1.1}
For $n=2$ and $\mathfrak{t}=\mathfrak{t}_{d}$, the $({\rm sl}_{2}(\mathbb{C}), \mathfrak{t}_{d})$-hierarchy is the AKNS-hierarchy, introduced in \cite{FNR}. They are the evolution equations of a deformation of the matrix  $\left(
\begin{smallmatrix}
-i& 0\\
0&i
\end{smallmatrix}
\right)$ and 
the simplest nontrivial equations in this hierarchy are the AKNS-equations. Recall that this is the 
system of differential equations for two complex functions $q$ and $r$, depending of the variables $x$ and $t$:
\begin{align}
\label{akns1}
&i\dfrac{\partial}{\partial t}q(x,t) :=  iq_t =-\dfrac{1}{2} q_{xx} + q^2r,\\ \notag
&i\dfrac{\partial}{\partial t}r(x,t) :=ir_t  =\dfrac{1}{2}r_{xx}-qr^2.
\end{align}
Ablowitz, Kaup, Newell and Segur showed in \cite{AKNS} that the initial value problem of (\ref{akns1}) could be solved with the Inverse
Scattering Transform, which explains the abbreviation.
\end{example}

The second choice for $\mathfrak{c}$ is the Lie subalgebra $C_{>0}$ of $C_{\geqslant 0}$ spanned by the elements $\{ E_{\alpha}z^{i} \mid i>0, 1 \leqslant \alpha  \leqslant r \}$. In this case we consider more general deformations of the generators $\{ E_{\alpha}z \mid 1 \leqslant \alpha  \leqslant r \}$ and their evolution equations involve now the projection of ${\rm sl_{n}}$-loops onto their part containing only strict positive powers of $z$. This brings you to the 
so-called {\it strict} $({\rm sl}_{n}(\mathbb{C}), \mathfrak{t})$-hierarchy. We follow here the terminology used in similar situations, see \cite{HHP} and \cite{H2016}. 

A detailed description of both integrable hierarchies and their properties can be found in the next section. 
In the third section we introduce, inspired by \cite{HHO11}, certain deformations of the Lie algebra 
$$
C=\{ \sum_{i \in \mathbb{Z}}\sum_{1 \leqslant \alpha  \leqslant r} t_{i \alpha} E_{\alpha}z^{i} \mid t_{i \alpha} \in \mathbb{C} \text{ and } t_{i \alpha} \neq 0 \text{ for a finite number of }i\}
$$ 
and a set of evolution equations they have to satisfy. It will be shown that this system can be seen as a merging of the $({\rm sl}_{n}(\mathbb{C}), \mathfrak{t})$-hierarchy and its strict version. Therefore we call it the {\it combined $({\rm sl}_{n}(\mathbb{C}), \mathfrak{t})$-hierarchy}.  Also this combined system turns out to be compatible. 
The subsequent section is devoted to the description of the linearization of the combined 
hierarchy, which is useful for the construction of its solutions. 
We conclude with giving a geometric construction, starting from a space of ${\rm sl_{n}}$-loops, of solutions of the combined $({\rm sl}_{n}(\mathbb{C}), \mathfrak{t})$-hierarchy.

\section{The $({\rm sl}_{n}(\mathbb{C}), \mathfrak{t})$-hierarchy and its strict version}
\label{Sversions}

We present here an algebraic description of the $({\rm sl}_{n}(\mathbb{C}), \mathfrak{t})$-hierarchy and its strict version that underlines the deformation character of these hierarchies as pointed out in the introduction. Recall that in the case of the AKNS-hierarchy, one worked with ${\rm sl}_{2}$-loops, where the coefficients of the powers of the loop parameter $z$ are ${\rm sl}_{2}$-matrices depending of various parameters. Here we discuss such ${\rm sl}_{n}$-loops.
We formalize this algebraically as follows: let $R$ be a commutative complex algebra that should be seen as the source from which the coefficients of the  $n \times n$ -matrices are taken. We will work in the Lie algebra ${\rm sl}_{n}(R)[z, z^{-1})$ consisting of all elements 
\begin{equation}
\label{lsl2}
X=\sum_{i=-\infty}^{N} X_{i}z^{i}, \text{ with all }X_{i} \in {\rm sl}_{n}(R)
\end{equation}
 and the bracket 
 $$
 [X,Y]=[\sum_{i=-\infty}^{N} X_{i}z^{i} , \sum_{j=-\infty}^{M} Y_{j}z^{j} ]:=\sum_{i=-\infty}^{N} \sum_{j=-\infty}^{M} [X_{i}, Y_{j}] z^{i+j}.
 $$
We will also make use of the slightly more general Lie algebra ${\rm gl}_{n}(R)[z, z^{-1})$, where the coefficients in the $z$-series from (\ref{lsl2}) are taken from ${\rm gl}_{n}(R)$ instead of ${\rm sl}_{n}(R)$ and the bracket is given by the same formula. In the Lie algebra ${\rm gl}_{n}(R)[z, z^{-1})$ we decompose elements in two ways. The first is as follows:
\begin{equation}
\label{eltdeco}
X=\sum_{i=-\infty}^{N} X_{i}z^{i}=\sum_{i=0}^{N} X_{i}z^{i} +\sum_{i=-\infty}^{-1} X_{i}z^{i}=:\pi_{\geqslant 0}(X)+\pi_{<0}(X)
\end{equation}
and this induces the splitting
\begin{equation}
\label{AKNSdeco}
{\rm gl}_{n}(R)[z, z^{-1})=\pi_{\geqslant 0}({\rm gl}_{n}(R)[z, z^{-1})) \oplus \pi_{<0}({\rm gl}_{n}(R)[z, z^{-1})), 
\end{equation}
where the two Lie subalgebras $\pi_{\geqslant 0}({\rm gl}_{n}(R)[z, z^{-1}))$ and $\pi_{<0}({\rm gl}_{n}(R)[z, z^{-1}))$ are given by
\begin{align*}
&\pi_{\geqslant 0}({\rm gl}_{n}(R)[z, z^{-1}))=\{ X \in {\rm gl}_{n}(R)[z, z^{-1}) \mid X=\pi_{\geqslant 0}(X)
\} \text{ and }\\
&\pi_{<0}({\rm gl}_{n}(R)[z, z^{-1}))=\{ X \in {\rm gl}_{n}(R)[z, z^{-1}) \mid X=\pi_{<0}(X)
\}.
\end{align*}
By restriction it leads to a similar decomposition for ${\rm sl}_{n}(R)[z, z^{-1})$, which is relevant for the $({\rm sl}_{n}(\mathbb{C}), \mathfrak{t})$-hierarchy. 
The second way to decompose elements of 
${\rm gl}_{n}(R)[z, z^{-1})$ is: 
\begin{equation}
\label{eltdeco2}
X=\sum_{i=-\infty}^{N} X_{i}z^{i}=\sum_{i=1}^{N} X_{i}z^{i} +\sum_{i=-\infty}^{0} X_{i}z^{i}=:\pi_{>0}(X)+\pi_{\leqslant 0}(X).
\end{equation}
This yields the splitting
\begin{equation}
\label{SAKNSdeco}
{\rm gl}_{n}(R)[z, z^{-1})=\pi_{>0}({\rm gl}_{n}(R)[z, z^{-1}))\oplus \pi_{\leqslant 0}({\rm gl}_{n}(R)[z, z^{-1})). 
\end{equation}
By restricting it to ${\rm sl}_{n}(R)[z, z^{-1})$, we get a similar decomposition for this Lie algebra, which relates, as we will see further on, to the strict version of the $({\rm sl}_{n}(\mathbb{C}), \mathfrak{t})$-hierarchy. The two Lie subalgebras $\pi_{>0}({\rm gl}_{n}(R)[z, z^{-1}))$ and $\pi_{\leqslant 0}({\rm gl}_{n}(R)[z, z^{-1}))$ in (\ref{SAKNSdeco}) are defined in a similar way as the first two Lie subalgebras
\begin{align*}
&\pi_{>0}({\rm gl}_{n}(R)[z, z^{-1}))=\{ X \in {\rm gl}_{n}(R)[z, z^{-1}) \mid X=\pi_{> 0}(X) 
\} \text{ and }\\
&\pi_{\leqslant 0}({\rm gl}_{n}(R)[z, z^{-1}))=\{ X \in {\rm gl}_{n}(R)[z, z^{-1}) \mid X=\pi_{\leqslant 0}(X)
\}.
\end{align*}

Next we describe the form of the deformations of the Lie algebras $C_{\geqslant 0}$ resp. $C_{> 0}$. We start with those of $C_{\geqslant 0}$. Note that 
$C_{\geqslant 0}=\pi_{\geqslant 0}(C)$ is a commutative algebra in the first component of the decomposition (\ref{AKNSdeco}) and our interest is in perturbations of $C_{\geqslant 0}$ obtained by conjugation with elements from a group connected to $\pi_{< 0}({\rm gl}_{n}(R)[z, z^{-1}))$, the complement of $\pi_{\geqslant 0}({\rm gl}_{n}(R)[z, z^{-1}))$ in (\ref{AKNSdeco}).
Note now that for each $X \in \pi_{<0}({\rm gl}_{n}(R)[z, z^{-1}))$ the exponential map yields a well-defined element of the form
\begin{equation}
\label{exp}
\exp(X)=\sum_{k=0}^{\infty} \frac{1}{k!}X^{k}=\Id +Y, Y \in \pi_{< 0}({\rm gl}_{n}(R)[z, z^{-1})),
\end{equation}
and with the formula for the logarithm one retrieves $X$ back from $Y$. One verifies directly that the elements of the form (\ref{exp}) form a group w.r.t. multiplication and this
we see as the group $G_{<0}$ corresponding to $\pi_{<0}({\rm gl}_{n}(R)[z, z^{-1}))$. Clearly, each Lie algebra $gC_{\geqslant 0}g^{-1}, g \in G_{<0},$ is commutative.

Now we discuss the shape of the deformations of $C_{>0}$. Note that 
$C_{> 0}=\pi_{> 0}(C)$ is a commutative algebra in the first component of the decomposition (\ref{SAKNSdeco}) and our interest is in perturbations of $C_{> 0}$ obtained by conjugation with elements from a group linked to $\pi_{\leqslant 0}({\rm gl}_{n}(R)[z, z^{-1}))$, the complement of $\pi_{> 0}({\rm gl}_{n}(R)[z, z^{-1}))$ in (\ref{SAKNSdeco}).
In the case of the Lie subalgebra $\pi_{\leqslant 0}({\rm gl}_{n}(R)[z, z^{-1}))$, one cannot move back and forth between the Lie algebra and its group. Nevertheless, one can assign a proper group to this Lie algebra. A priori, the exponential $\exp(Y)$ of an element $Y \in \pi_{\leqslant 0}({\rm gl}_{n}(R)[z, z^{-1}))$, does not have to define an element in $\pi_{\leqslant 0}({\rm gl}_{n}(R)[z, z^{-1}))$. That requires convergence conditions. However, if it does, then it belongs to  
$$
G_{\leqslant 0}=\{ K=\sum_{j=0}^{\infty}K_{j}z^{-j} \mid \text{ all } K_{j}\in {\rm gl}_{n}(R), K_{0} \in {\rm gl}_{n}(R)^{*} \},
$$
where ${\rm gl}_{n}(R)^{*}$ denotes the elements in ${\rm gl}_{n}(R)$ that have a multiplicative inverse in ${\rm gl}_{n}(R)$. It is a direct verification that $G_{\leqslant 0}$ is a group and we see it as a proper group corresponding to the Lie algebra $\pi_{\leqslant 0}({\rm gl}_{n}(R)[z, z^{-1}))$. In fact, as a group $G_{\leqslant 0}$ is isomorphic to the semi-direct 
product of $G_{<0}$ and ${\rm gl}_{n}(R)^{*}$. Note that also every Lie algebra $gC_{> 0}g^{-1}, g \in G_{\leqslant 0},$ is commutative.

Despite of the fact that the deformations of $C_{\geqslant 0}$ and $C_{> 0}$ described above, involve  ${\rm gl}_{n}(R)$-loops, the deformed Lie algebras 
remain in ${\rm sl}_{n}(R)[z, z^{-1})$.
For there holds
\begin{lemma}
\label{L2.1}
The group $G_{\leqslant 0}$ acts by conjugation on ${\rm sl}_{n}(R)[z, z^{-1})$.
\end{lemma}
\begin{proof}
Take first any $g \in G_{<0}$. Then there is an $X \in \pi_{<0}({\rm gl}_{n}(R)[z, z^{-1}))$ such that $g=\exp(X)$. Since there holds for every $Y \in {\rm sl}_{n}(R)[z, z^{-1})$
$$
gYg^{-1}=\exp(X)Y \exp(-X)=e^{\text{ad}(X)}(Y)=Y+ \sum_{k=1}^{\infty} \frac{1}{k!} \text{ad}(X)^{k}(Y)
$$
and this shows that the coefficients for the different powers of $z$ in this expression are commutators of elements of ${\rm gl}_{n}(R)$ and ${\rm sl}_{n}(R)$ and that proofs the claim for elements from $G_{<0}$. Since conjugation with an element from ${\rm gl}_{n}(R)^{*}$ maps ${\rm sl}_{n}(R)$ to itself, the same holds for ${\rm sl}_{n}(R)[z, z^{-1})$. This proves the full claim.
\end{proof}

Next we have a more detailed look at the different deformations.
Since the elements $z^{m} \Id, m \in \mathbb{Z},$ are central in ${\rm gl}_{n}(R)[z, z^{-1})$, the deformations of $C_{\geqslant 0}$ by elements of $G_{<0}$ are determined by that of the $\{ E_{\alpha} \}$:
$$
gE_{\alpha}z^{m}g^{-1}=gE_{\alpha}g^{-1}z^{m}, g \in G_{<0}, m \geqslant 0.
$$ 
Therefore we focus for $g=\exp(X)=\exp(\sum_{j=1}^{\infty}X_{j}z^{-j}) \in G_{<0}$ on the $\{ gE_{\alpha}g^{-1} \}$. These deformed generators have the form 
\begin{align}
\label{decoU}
\notag
U_{\alpha}&=gE_{\alpha}g^{-1}=\exp(X)E_{\alpha} \exp(-X):=\sum_{i=0}^{\infty}U_{\alpha,i}z^{-i}\\
&=E_{\alpha}
+ [X_{1},E_{\alpha}]z^{-1} +\{ [X_{2},E_{\alpha}] + \frac{1}{2} [X_{1}, [X_{1}, E_{\alpha}]]\}z^{-2} +\cdots 
\end{align}
\begin{remark}
\label{R2.1} In the case of the $({\rm sl}_{2}(\mathbb{C}), \mathfrak{t}_{d})$-hierarchy we only have to deal with the deformation $U_{1}$ of the matrix $E_{1}=\left(
\begin{smallmatrix}
-i& 0\\
0&i
\end{smallmatrix}
\right)$.
From formula (\ref{decoU}) we see directly that, if each $X_{i}=\left(
\begin{matrix}
-\alpha_{i}& \beta_{i}\\
\gamma_{i}&\alpha_{i}
\end{matrix}
\right), i=1,2,$ then $U_{1,0}=E_{1}$,
\begin{equation}
\label{coefQ1}
U_{1,1}:=\left(
\begin{matrix}
0&q \\
r&0
\end{matrix}
\right)=\left(
\begin{matrix}
0& 2i \beta_{1}\\
-2i \gamma_{i}&0
\end{matrix}
\right),
\end{equation}
and
\begin{equation}
\label{}
U_{1,2}:=\left(
\begin{matrix}
u_{11}& u_{12}\\
u_{21}&u_{22}
\end{matrix}
\right)=\left(
\begin{matrix}
-2i \beta_{1} \gamma_{1}& 2i(\beta_{2}-\alpha_{1}\beta_{1}) \\
-2i(\gamma_{2}+\alpha_{1}\gamma_{1})&2i \beta_{1} \gamma_{1}
\end{matrix}
\right).
\end{equation}
In particular we get in this way that $u_{11}=-i\frac{qr}{2}$ and $u_{22}=i\frac{qr}{2}$. We will see at a later stage that the deformation equations imply that also the coefficients $u_{12}$ and $u_{21}$ can expressed in $q$, $u$ and their derivatives.
\end{remark}
By the same argument, the deformations of the Lie algebra $C_{>0}$ by elements from $G_{\leqslant 0}$ are basically determined by those of the elements $\{E_{\alpha}z\}$. So we focus on the deformations of these elements. Using the same notations as at the deformation of $E_{\alpha}$ by $G_{<0}$, we get that the deformation of $E_{\alpha} z$ by a $Kg \in G_{\leqslant 0}$, with $K \in {\rm gl}_{n}(R)^{*}$ and $g \in G_{<0}$, looks like
\begin{align}
\label{sdecoV}
\notag
V_{\alpha}&=KgE_{\alpha}zg^{-1}K^{-1}:=\sum_{i=0}^{\infty}KU_{\alpha,i}K^{-1}z^{1-i}= \sum_{i=0}^{\infty}V_{\alpha, i}z^{1-i}\\
&=
V_{\alpha,0}z+ [KX_{1}K^{-1},V_{\alpha,0}] +
\cdots .
\end{align}
Consequently, the corresponding deformation of each $E_{\alpha} z^{m}, m \geqslant 1,$ is $V_{\alpha}z^{m-1}$.

Having fixed the type of deformation of the Lie algebras $C_{\geqslant 0}$ and $C_{> 0}$, we need one more ingredient to discuss the hierarchies.
In the case of deforming $C_{\geqslant 0}$, all the basis elements 
$\{ E_{\alpha}z^{m}, m \geqslant 0, 1 \leqslant \alpha \leqslant r \}$ generate commuting flows and we want to study deformations of the type (\ref{decoU}) that depend of these flows.
Therefore we assume that the algebra $R$ possesses a set $\{ \partial_{m \alpha} \mid m \geqslant 0, 1 \leqslant \alpha \leqslant r \}$ of commuting $\mathbb{C}$-linear derivations $\partial_{m \alpha}: R \to R$, where each $\partial_{m \alpha}$ should be seen as the derivation corresponding to the flow generated by $E_{\alpha}z^{m}$. The data $(R, \{ \partial_{m \alpha } \mid m \geqslant 0, 1 \leqslant \alpha \leqslant r \})$ is called a {\it setting} for the 
integrable hierarchy related to the perturbation of the Lie algebra $C_{\geqslant 0}$ to be introduced in a moment.
Similarly, at deforming $C_{> 0}$, we require that $R$ has a collection 
of commuting $\mathbb{C}$-linear  derivations $\partial_{m \alpha}: R \to R,$ for the indices $m\geqslant 1, 1 \leqslant \alpha \leqslant r.$  Sticking to the same terminology, we call the data $(R, \{ \partial_{m \alpha} \mid m \geqslant 1 ,1 \leqslant \alpha \leqslant r\})$ also a {\it setting} for the integrable hierarchy corresponding to the deformations of $C_{> 0}$ to be discussed here.
\begin{example}
\label{E2.1}
Examples of settings for the respective hierarchies are the algebras of complex polynomials $\mathbb{C}[t_{m\alpha}]$ in the variables $\{ t_{m \alpha} \mid m \geqslant 0, 1 \leqslant \alpha \leqslant r\}$ resp. $\{ t_{m\alpha} \mid m \geqslant 1, 1 \leqslant \alpha \leqslant r \}$ or the formal power series $\mathbb{C}[[t_{m\alpha}]]$ in the same variables, where both algebras are equipped with the derivations $\partial_{m\alpha}=\frac{\partial}{\partial t_{m\alpha}}$ for the appropriate indices. Depending if one takes $R$ equal to $\mathbb{C}[t_{m\alpha}]$ or $\mathbb{C}[[t_{m\alpha}]]$, one speaks of  the polynomial or formal power series solutions of the hierarchies. 
\end{example}
We let each derivation $\partial_{m\alpha}$, occurring in some setting, act coefficient wise on the matrices from ${\rm gl}_{n}(R)$ and that defines then a derivation of this algebra. The same holds for the extension to ${\rm gl}_{n}(R)[z, z^{-1})$ defined by 
$$
\partial_{m\alpha}(X):=\sum_{j=-\infty}^{N} \partial_{m \alpha}(X_{j})z^{j}.
$$
Now it is time to discuss the nonlinear equations that the deformations $\{ U_{\alpha_{2}} \}$ of type (\ref{decoU})
should satisfy. We want that  
their evolution w.r.t. the $\{ \partial_{m \alpha_{1}} \}$ is given by: for all $m \geqslant 0$ and all $\alpha_{1}, 1 \leqslant \alpha_{1} \leqslant r ,$
\begin{equation}
\label{LaxAKNS}
\partial_{m\alpha_{1}}(U_{\alpha_{2}})=[\pi_{ \geqslant 0}(U_{\alpha_{1}}z^{m}),U_{\alpha_{2}}]=-[\pi_{ < 0}(U_{\alpha_{1}}z^{m}),U_{\alpha_{2}}],
\end{equation}
where the second identity follows from the fact that all $\{ U_{\alpha_{1}}z^{m} \}$ commute.
The equations (\ref{LaxAKNS}) are called the {\it Lax equations of the $({\rm sl}_{n}(\mathbb{C}), \mathfrak{t})$-hierarchy} and the deformation $\{ U_{\alpha_{2}} \}$ satisfying these equations is called a {\it solution} of the hierarchy. Note that the $\{U_{\alpha}=E_{\alpha}\}$ form a solution of the $({\rm sl}_{n}(\mathbb{C}), \mathfrak{t})$-hierarchy and it is called the {\it trivial} one.
Note that the equations (\ref{LaxAKNS}) for $m=0$ are simply $\partial_{0\alpha_{1}}(U_{\alpha_{2}})=[E_{\alpha_{1}}, U_{\alpha_{2}}]$. Therefore, if $\partial_{0\alpha_{1}}=\frac{\partial}{\partial t_{0\alpha_{1}}}$ and the matrix coefficients of both $\exp(\sum_{1 \leqslant \alpha_{1} \leqslant r}t_{0\alpha_{1}}E_{\alpha_{1}})$ and its inverse
belong to the algebra $R$ of matrix coefficients, then we can introduce the deformation $\hat{U}_{\alpha_{2}}$ given by
$$
\hat{U}_{\alpha_{2}}:=\exp(-\sum_{1 \leqslant \alpha_{1} \leqslant r}t_{0\alpha_{1}}E_{\alpha_{1}})U_{\alpha_{2}} \exp(\sum_{1 \leqslant \alpha_{1} \leqslant r}t_{0\alpha_{1}}E_{\alpha_{1}}),
$$
which is easily seen to satisfy $\partial_{0\alpha_{1}}(\hat{U}_{\alpha_{2}})=0$, for all $\alpha_{1}$. This handles then the dependence of $U_{\alpha_{2}}$ of the $\{t_{0\alpha_{1}}\}$.

For the deformations $\{ V_{\beta_{2}} \}$ of the form (\ref{sdecoV}) we require that the evolution w.r.t. the $\{ \partial_{m\beta_{1}}\}$ is coupled to the decomposition (\ref{SAKNSdeco}) and it should satisfy: for all $m \geqslant 1$ and all $\beta_{1}, 1 \leqslant \beta_{1} \leqslant r ,$
\begin{equation}
\label{LaxSAKNS}
\partial_{m\beta_{1}}(V_{\beta_{2}})=[\pi_{ > 0}(V_{\beta_{1}}z^{m-1}),V_{\beta_{2}}]=-[\pi_{ \leqslant 0}(V_{\beta_{1}}z^{m-1}),V_{\beta_{2}}],
\end{equation}
where the second identity follows from the fact that all $\{ V_{\beta_{1}} z^{m-1} \}$ commute.
Since the equations (\ref{LaxSAKNS}) correspond to the strict cut-off (\ref{eltdeco2}), they are called the {\it Lax equations of the strict $({\rm sl}_{n}(\mathbb{C}), \mathfrak{t})$-hierarchy} and any set of deformation $\{ V_{\beta_{2}} \}$ satisfying them, is called a {\it solution} of this hierarchy. Again there is al least one solution: $V_{\beta_{1}}=E_{\beta_{1}}z, $ for all $\beta_{1}, 1 \leqslant \beta_{1} \leqslant r $. It is called the {\it trivial} solution of the hierarchy.
\begin{remark}
\label{R2.2} Assume $\mathfrak{t}_{1}$ and $\mathfrak{t}_{2}$ are conjugated under $\GL_{n}(\mathbb{C})$, i.e. there is a $g \in \GL_{n}(\mathbb{C})$ such that $\mathfrak{t}_{2}=g \mathfrak{t}_{1}g^{-1}$. Then one verifies directly that if the $\{ U_{\alpha} \}$ is a solution of the $({\rm sl}_{n}(\mathbb{C}), \mathfrak{t}_{1})$-hierarchy or the $\{ V_{\beta} \}$ of its strict version, then the  $\{ gU_{\alpha}g^{-1} \}$ solve the Lax equations of the $({\rm sl}_{n}(\mathbb{C}), \mathfrak{t}_{2})$-hierarchy and the $\{ gV_{\beta}g^{-1} \}$ those of the strict $({\rm sl}_{n}(\mathbb{C}), \mathfrak{t}_{2})$-hierarchy and in both cases this forms a bijection between the two sets of solutions. Hence, it suffices to consider the systems for one representative of each conjugacy class. In particular, we may suppose that $\mathfrak{t}$ is upper triangular.
\end{remark}

For both systems (\ref{LaxAKNS}) and (\ref{LaxSAKNS}) one can speak of compatibility. There holds namely
\begin{proposition}
\label{P2.1}
Both sets of Lax equations (\ref{LaxAKNS}) and (\ref{LaxSAKNS}) are so-called compatible systems, i.e. the projections $\{ B_{m\alpha}:=\pi_{\geqslant 0}(U_{\alpha}z^{m}) \mid m \geqslant 0, 1 \leqslant \alpha \leqslant r \}$ satisfy the zero curvature relations
\begin{equation}
\label{ZCAKNS}
\partial_{m_{1}\alpha_{1}}(B_{m_{2}\alpha_{2}})-\partial_{m_{2}\alpha_{2}}(B_{m_{1}\alpha_{1}})-[B_{m_{1}\alpha_{1}},B_{m_{2}\alpha_{2}}]=0
\end{equation}
and the projections $\{ C_{m\beta}:=\pi_{> 0}(V_{\beta}z^{m-1}) \mid m \geqslant 1 , 1 \leqslant \beta \leqslant r \}$ satisfy the zero curvature relations
\begin{equation}
\label{ZCSAKNS}
\partial_{m_{1}\beta_{1}}(C_{m_{2}\beta_{2}})-\partial_{m_{2}\beta_{2}}(C_{m_{1}\beta_{1}})-[C_{m_{1}\beta_{1}},C_{m_{2}\beta_{2}}]=0
\end{equation}
\end{proposition}
\begin{proof}
The idea is to show that the left hand side of (\ref{ZCAKNS}) resp. (\ref{ZCSAKNS}) belongs to 
$$
{\rm sl}_{n}(R)[z, z^{-1})_{\geqslant 0} \cap {\rm sl}_{n}(R)[z, z^{-1})_{<0} \text{ resp. }
{\rm sl}_{n}(R)[z, z^{-1})_{> 0} \cap {\rm sl}_{n}(R)[z, z^{-1})_{\leqslant 0}
$$ 
and thus has to be zero. We give the proof for the $\{ C_{m\beta} \}$, that for the $\{ B_{m\alpha} \}$ is similar and is left to the reader. The inclusion in the first factor is clear as both $C_{m\beta}$ and $\partial_{n\gamma}(C_{m\beta})$ belong to the Lie subalgebra ${\rm sl}_{n}(R)[z, z^{-1})_{ > 0}.$ 
To show the other one, we use the Lax equations (\ref{LaxSAKNS}). 
Note that the same Lax equations hold for all the $\{ z^{N}V_{\beta} \mid N \geqslant 0 \}$
\begin{equation*}
\partial_{m\beta_{1}}(z^{N}V_{\beta})=[\pi_{ > 0}(V_{\beta_{1}}z^{m-1}),z^{N}V_{\beta}].
\end{equation*}
By substituting $C_{m_{i}\beta_{i}}=z^{m_{i}-1}V_{\beta_{i}}-\pi_{<0}(z^{m_{i}-1}V_{\beta_{i}})$ we get for
\begin{align*}
\partial_{m_{1}\beta_{1}}(C_{m_{2}\beta_{2}})-\partial_{m_{2}\beta_{2}}(C_{m_{1}\beta_{1}})=&\;\partial_{m_{1}\beta_{1}}(z^{m_{2}-1}V_{\beta_{2}})-\partial_{m_{1}\beta_{1}}(\pi_{\leqslant 0}(z^{m_{2}-1}V_{\beta_{2}}))\\
 & -\partial_{m_{2}\beta_{2}}(z^{m_{1}-1}V_{\beta_{1}})+\partial_{m_{2}\beta_{2}}(\pi_{\leqslant 0}(z^{m_{1}-1}V_{\beta_{1}}))\\
=&\; [C_{m_{1}\beta_{1}},z^{m_{2}-1}V_{\beta_{2}}]-[C_{m_{2}\beta_{2}},z^{m_{1}-1}V_{\beta_{1}}]\\
& -\partial_{m_{1}\beta_{1}}(\pi_{\leqslant 0}(z^{m_{2}-1}V_{\beta_{2}}))+\partial_{m_{2}\beta_{2}}(\pi_{\leqslant 0}(z^{m_{1}-1}V_{\beta_{1}}))
\end{align*} 
and for 
\begin{align*}
[C_{m_{1}\beta_{1}},C_{m_{2}\beta_{2}}]=&\;[z^{m_{1}-1}V_{\beta_{1}}-\pi_{\leqslant 0}(z^{m_{1}-1}V_{\beta_{1}}), z^{m_{2}-1}V_{\beta_{2}}-\pi_{\leqslant 0}(z^{m_{2}-1}V_{\beta_{2}})]\\
=&\;-[\pi_{\leqslant 0}(z^{m_{1}-1}V_{\beta_{1}}), z^{m_{2}-1}V_{\beta_{2}}]+[\pi_{\leqslant 0}(z^{m_{2}-1}V_{\beta_{2}}), z^{m_{1}-1}V_{\beta_{1}}]\\
&\;+[\pi_{\leqslant 0}(z^{m_{1}-1}V_{\beta_{1}}),\pi_{\leqslant 0}(z^{m_{2}-1}V_{\beta_{2}})].
\end{align*}
Taking into account the second identity in (\ref{LaxSAKNS}), we see that the left hand side of (\ref{ZCSAKNS}) is equal to 
$$
-\partial_{m_{1}\beta_{1}}(\pi_{\leqslant 0}(z^{m_{2}-1}V_{\beta_{2}}))+\partial_{m_{2}\beta_{2}}(\pi_{\leqslant 0}(z^{m_{1}-1}V_{\beta_{1}}))-[\pi_{\leqslant 0}(z^{m_{1}-1}V_{\beta_{1}}),\pi_{\leqslant 0}(z^{m_{2}-1}V_{\beta_{2}})].
$$
This element belongs to the Lie subalgebra $\pi_{\leqslant 0}({\rm sl}_{n}(R)[z, z^{-1}))$ and that proves the claim.
\end{proof}
Reversely, we have
\begin{proposition} 
\label{P2.2}
Suppose we have a deformation $\{U_{\alpha}\}$ of the type (\ref{decoU}) and a deformation $\{V_{\beta}\}$ of the form (\ref{sdecoV}). Then there holds:
\begin{enumerate}
\item Assume that the projections $\{ B_{m \alpha}:=\pi_{\geqslant 0}(U_{\alpha}z^{m}) \mid m \geqslant 0, 1 \leqslant \alpha \leqslant r \}$ satisfy the zero curvature relations (\ref{ZCAKNS}). Then the set $\{U_{\alpha}\}$ is a solution of the $({\rm sl}_{n}(\mathbb{C}), \mathfrak{t})$-hierarchy.
\item Similarly, if the projections $\{ C_{m\beta}:=\pi_{> 0}(V_{\beta}z^{m-1}) \mid m \geqslant 1,1 \leqslant \beta \leqslant r  \}$ satisfy the zero curvature relations (\ref{ZCSAKNS}), then the set $\{V_{\beta}\}$ is a solution of the strict $({\rm sl}_{n}(\mathbb{C}), \mathfrak{t})$-hierarchy.
\end{enumerate}
\end{proposition}
 \begin{proof}
 Again we prove the statement for the set $\{V_{\beta}\}$, that for the $\{U_{\alpha}\}$ is shown in a similar way.
 So, assume that there is one Lax equation (\ref{LaxSAKNS}) that does not hold. Then there is a $m \geqslant 1$, a $\beta_{1}, 1 \leqslant \beta_{1} \leqslant r$ and a $\beta_{2}, 1 \leqslant \beta_{2} \leqslant r$ such that 
$$
\partial_{m\beta_{1}}(V_{\beta_{2}})-[C_{m\beta_{1}},V_{\beta_{2}}]=\sum_{j \leqslant k(m,\beta_{i})} X_{j}z^{j}, \text{ with }X_{k(m,\beta_{i})} \neq 0.
$$
Since both $\partial_{m \beta_{1}}(V_{\beta_{2}})$ and $[C_{m \beta},V_{\beta_{2}}]$ are of order smaller than or equal to one in $z$, we know that $k(m,\beta_{i}) \leqslant 0$. Further, we can say that for all 
$N\geqslant 0 $
$$
\partial_{m\beta_{1}}(z^{N}V_{\beta_{2}})-[C_{m\beta_{1}},z^{N}V_{\beta_{2}}]=\sum_{j \leqslant k(m,\beta_{i})} X_{j}z^{j+N}, \text{ with }X_{k(m,\beta_{i})} \neq 0
$$
and we see by letting $N$ go to infinity that the right hand side can obtain any sufficiently large order in $z$. By the zero curvature relation for the indices $N\beta_{2}$ and $m\beta_{1}$ we get for the left hand side
\begin{align*}
\partial_{m\beta_{1}}(z^{N}V_{\beta_{2}})-[C_{m\beta_{1}},z^{N}V_{\beta_{2}}]&=\partial_{m}(C_{N\beta_{2}}) -[C_{m\beta_{1}}, C_{N\beta_{2}}]+\partial_{m\beta_{1}}(\pi_{\leqslant 0}(z^{N}V_{\beta_{2}}))\\
&\; \;-[C_{m\beta_{1}}, \pi_{\leqslant 0}(z^{N}V_{\beta_{2}})]\\
&=\partial_{N}(C_{m\beta_{1}})+ \partial_{m\beta_{1}}(\pi_{ \leqslant 0}(z^{N}V_{\beta_{2}})) \\
&\; \;-[C_{m\beta_{1}}, \pi_{ \leqslant 0}(z^{N}V_{\beta_{2}})]
\end{align*}
and this last expression is of order smaller or equal to $m$ in $z$. This contradicts the unlimited growth in orders of $z$ of the right hand side. Hence all Lax equations (\ref{LaxSAKNS}) have to hold for $V_{\beta_{2}}$.
\end{proof}
Because of the equivalence between the Lax equations (\ref{LaxAKNS}) for the $\{U_{\alpha}\}$ and the zero curvature relations (\ref{ZCAKNS}) for the $\{ B_{m\alpha} \}$, we call this last set of equations also the {\it zero curvature form} of the $({\rm sl}_{n}(\mathbb{C}), \mathfrak{t})$-hierarchy. Similarly, the zero curvature relations (\ref{ZCSAKNS}) for the $\{ C_{m\beta} \}$ is called the {\it zero curvature form} of the strict $({\rm sl}_{n}(\mathbb{C}), \mathfrak{t})$-hierarchy.

\begin{remark}
\label{R2.3} 
We come back to example (\ref{R2.1}). Assume that the deformation $U_{1}$ of $E_{1}$ discussed there is a solution of the $({\rm sl}_{2}(\mathbb{C}), \mathfrak{t}_{d})$-hierarchy. We have seen that the dependence of the parameter $t_{01}$ is straightforward. Therefore the first serious relation occurs for $m_{1}=2, m_{2}=1$ and $\alpha_{1}=\alpha_{2}=1$. Consider the relation (\ref{ZCAKNS}) for these indices
$$
\partial_{21}(E_{1}z +U_{1,1})=\partial_{11}(E_{1}z^{2} +U_{1,1}z +U_{1,2})+[U_{1,2}, E_{1}z +U_{1,1}].
$$
Since $E_{1}$ is constant, this identity reduces in ${\rm sl}_{2}(R)[z, z^{-1})_{ \geqslant 0}$ to the following two equalities
\begin{align}
\label{ZC12}
\partial_{11}(U_{1,1})=[E_{1}, U_{1,2}] \text{ and }\partial_{21}(U_{1,1}) =\partial_{11}(U_{1,2})+[U_{1,2},U_{1,1}].
\end{align}
The first gives an expression of the off-diagonal terms of $U_{1,2}$ in the coefficients $q$ and $r$ of $U_{1,1}$, i.e.
$$
u_{12}=\frac{i}{2}\partial_{11}(q) \text{ and }q_{21}=-\frac{i}{2}\partial_{11}(r).
$$ 
Then the second equation becomes a system of equations solely in the coefficients $q$, $r$ and their derivatives w.r.t. $\partial_{11}$ and $\partial_{21}$. A direct computation shows that it amounts to the AKNS-equations (\ref{akns1}), if one has $\partial_{11}=\frac{\partial}{\partial x}$ and $\partial_{21}=\frac{\partial}{\partial t}$.
\end{remark}

Besides the zero curvature relations for the cut-off's $\{ B_{m\alpha}\}$ resp. $\{ C_{m\beta}\}$ corresponding to respectively a solution $\{U_{\alpha}\}$ of the $({\rm sl}_{n}(\mathbb{C}), \mathfrak{t})$-hierarchy and a solution $\{V_{\beta}\}$ of the strict $({\rm sl}_{n}(\mathbb{C}), \mathfrak{t})$-hierarchy, also other parts satisfy such relations. Define for all $\alpha, 1 \leqslant \alpha \leqslant r,$ and all $\beta, 1 \leqslant \beta \leqslant r,$
$$
A_{m\alpha}:=B_{m\alpha}-U_{\alpha}z^{m}, m \geqslant 0, \text{ and }D_{m\beta}:=C_{m\beta}-V_{\beta}z^{m-1}, m\geqslant 1.
$$
Then we can say

\begin{corollary}The following relations hold:
\label{C2.1}
\begin{itemize}
\item The parts $\{ A_{m\alpha} \mid m \geqslant 0, 1 \leqslant \alpha \leqslant r\}$ of a solution $\{U_{\alpha}\}$ of the $({\rm sl}_{n}(\mathbb{C}), \mathfrak{t})$-hierarchy satisfy 
$$
\partial_{m_{1}\alpha_{1}}(A_{m_{2}\alpha_{2}})-\partial_{m_{2}\alpha_{2}}(A_{m_{1}\alpha_{1}})-[A_{m_{1}\alpha_{1}},A_{m_{2}\alpha_{2}}]=0.
$$
\item
The parts $\{ D_{m\beta} \mid m \geqslant 1,  1 \leqslant \beta \leqslant r\}$ of a solution $\{V_{\beta}\}$ of the strict $({\rm sl}_{n}(\mathbb{C}), \mathfrak{t})$-hierarchy satisfy 
$$
\partial_{m_{1}\beta_{1}}(D_{m_{2}\beta_{2}})-\partial_{m_{2}\beta_{2}}(D_{m_{1}\beta_{1}})-[D_{m_{1}\beta_{1}},D_{m_{2}\beta_{2}}]=0
$$
\end{itemize}
\end{corollary}
\begin{proof}
Again we show the result only in the strict case. Recall that the $\{V_{\beta_{2}}z^{m_{2}-1}\}$ satisfy Lax equations similar to the $\{V_{\beta_{2}}\}$
$$
\partial_{m_{1}\beta_{1}}(V_{\beta_{2}}z^{m_{2}-1})=[D_{m_{1}\beta_{1}}, V_{\beta_{2}}z^{m_{2}-1}], i \geqslant 1.
$$
Now we substitute in the zero curvature relations for the $\{ C_{m\beta}\}$ everywhere the relation $C_{m\beta}=D_{m\beta}+V_{\beta}z^{m-1}$ and use the above Lax equations and the fact that all the $\{V_{\beta}z^{m-1}\}$ commute. This yields the desired result.
\end{proof}

\section{The combined $({\rm sl}_{n}(\mathbb{C}), \mathfrak{t})$-hierarchy}
\label{SAKNS+AKNS}

The commutative Lie subalgebra $C$, where the combined hierarchy is based upon, is the complex algebra with basis $\{ E_{\alpha}z^{m}\mid m \in \mathbb{Z}, 1 \leqslant \alpha \leqslant r\}$. It is a Lie subalgebra of both ${\rm sl}_{n}(R)[z, z^{-1})$ and ${\rm sl}_{n}(R)[ z^{-1}, z)$, where this last Lie algebra consists of loops that have at most a pole around zero:
$$
{\rm sl}_{n}(R)[ z^{-1}, z)=\{ \sum_{i=-N}^{\infty} X_{i}z^{i} \mid  \text{ all }X_{i} \in {\rm sl}_{n}(R)\}.
$$
The algebra $C$ can be split into $C=C_{\geqslant 0} \oplus C_{<0}$, where $C_{\geqslant 0}$ is spanned by the $\{ E_{\alpha}z^{m}\mid m \geqslant 0, 1 \leqslant \alpha \leqslant r\}$ and $C_{<0}$ by the 
$\{ E_{\alpha}z^{m}\mid m <0\}$. Now we are interested in deforming both $C_{\geqslant 0}$ and $C_{>0}$, the first inside ${\rm sl}_{n}(R)[z, z^{-1})$ and the second inside ${\rm sl}_{n}(R)[ z^{-1}, z)$. This may cause that the deformations of $C_{\geqslant 0}$ and $C_{>0}$ no longer commute. Since the powers of $z$ are central, it is enough to consider the deformations of the elements $\{E_{\alpha} \mid  1 \leqslant \alpha \leqslant r \}$ and the $\{E_{\alpha}z^{-1} \mid  1 \leqslant \alpha \leqslant r\}$. We deform the elements $\{E_{\alpha} \}$ as in the $({\rm sl}_{n}(\mathbb{C}), \mathfrak{t})$-case with an element of the group 
$$
G_{<0}=\{ \Id +Y_{<0} \mid Y_{<0} \in {\rm gl}_{n}(R)[z, z^{-1})_{<0} \}
$$
and that leads to a collection of deformations $U_{\alpha}=U_{\alpha}(z)=\sum_{j \geqslant 0}^{\infty} U_{\alpha j}z^{-j}$ as in (\ref{decoU}). The elements $\{E_{\alpha}z^{-1}\}$, on the contrary, we deform with an element from the group 
$$
G_{\geqslant 0}=\{ X=X_{0} + X_{\geqslant 1} \mid X_{0} \in {\rm gl}_{n}(R)^{*}, X_{\geqslant 1} \in {\rm gl}_{n}(R)[z^{-1}, z)_{> 0} \}
$$
to elements 
\begin{equation}
\label{defS}
W_{\alpha}=W_{\alpha}(z):=XE_{\alpha}z^{-1}X^{-1}=\sum_{j=0}^{\infty}S_{j}z^{j-1} \in {\rm sl}_{n}(R)[ z^{-1}, z).
\end{equation}
If one makes in $W_{\alpha}(z)$ the substitution $z \to z^{-1}$, then one gets deformations $V_{\alpha}(z)=W_{\alpha}(\dfrac{1}{z})$ as considered in (\ref{sdecoV}) at the strict $({\rm sl}_{n}(\mathbb{C}), \mathfrak{t})$-hierarchy. So, after deforming the basis of $C$, we are left with two sets $\{ U_{\alpha}z^{m}\mid m \geqslant 0 , 1 \leqslant \alpha \leqslant r\}$ and $\{W_{\alpha}z^{m+1}\mid m<0, 1 \leqslant \alpha \leqslant r \}$ that each span a commutative Lie algebra, but do not have to commute among each other.

Next we discuss the Lax equations that the deformation $(\{U_{\alpha}\}, \{W_{\beta}\})$ of the basis of $C$ should satisfy. Thereto we assume that the algebra $R$ possesses a collection $\{ \partial_{m \alpha} \mid m \in \mathbb{Z}, 1 \leqslant \alpha \leqslant r\}$ of commuting $\mathbb{C}$-linear derivations $\partial_{m \alpha}: R \to R$, where each $\partial_{m \alpha}$ should be seen as an algebraic substitute for the derivation corresponding to the flow generated by each element $E_{\alpha}z^{m}$
in the basis of $C$. For $X \in {\rm gl}_{n}(R)[z, z^{-1})$ or $X \in {\rm gl}_{n}(R)[z^{-1},z)$ we define 
the action of each $\partial_{m \alpha}$ by
$$
\partial_{m\alpha}(X):=\sum_{j} \partial_{m\alpha}(X_{j})z^{j},
$$
where the action on ${\rm gl}_{n}(R)$ is defined coefficient wise. This defines a derivation of both algebras. Following the terminology used in Section 2
, we call the data $(R, \{ \partial_{m\alpha} \mid m \in \mathbb{Z},1 \leqslant \alpha \leqslant r \})$ a {\it setting} for the {\it combined $({\rm sl}_{n}(\mathbb{C}), \mathfrak{t})$-hierarchy}.
\begin{example}
\label{E3.1}
Examples of settings are for the moment the algebras of complex polynomials $\mathbb{C}[t_{m \alpha}]$ in the variables $\{ t_{m \alpha} \mid m \in \mathbb{Z}\}$ 
or the formal power series $\mathbb{C}[[t_{m \alpha}]]$ in the same variables; both algebras equipped with the derivations $\partial_{m\alpha}=\frac{\partial}{\partial t_{m\alpha}}, m \in \mathbb{Z}, 1 \leqslant \alpha \leqslant r$. Later, at the construction of the solutions, more sophisticated choices for $R$ will occur.
\end{example}

Now we want that the sets of deformations $(\{U_{\alpha}\}, \{W_{\beta}\})$ satisfy the following evolution equations:
\begin{align}
\label{Lax1combAKNS}
&\partial_{m\alpha_{1}}(U_{\alpha_{2}})=[\pi_{ \geqslant 0}(U_{\alpha_{1}}z^{m}),U_{\alpha_{2}}] \text{ and } \partial_{m\alpha_{1}}(W_{\beta})=[\pi_{ \geqslant 0}(U_{\alpha_{1}}z^{m}),W_{\beta}], \\ 
\notag 
&\text{ for all $m \geqslant 0$ and all $\{\alpha_{i}\}, \beta \in [1,r] \cap \mathbb{Z}$;}\\
\label{Lax2combAKNS}
&\partial_{m\beta_{1}}(W_{\beta_{2}})=[\pi_{ < 0}(W_{\beta_{1}}z^{m+1}),W_{\beta_{2}}] \text{ and } \partial_{m\beta_{1}}(U_{\alpha})=[\pi_{ < 0}(W_{\beta_{1}}z^{m+1}),U_{\alpha}],\\ \notag
&\text{ for all $m < 0$ and all $\{\beta_{j}\}, \alpha \in [1,r] \cap \mathbb{Z}$.} 
\end{align}
 Note that the first set of equations in (\ref{Lax1combAKNS}) imply that the $\{U_{\alpha} \}$ satisfy the Lax equations of the $({\rm sl}_{n}(\mathbb{C}), \mathfrak{t})$-hierarchy w.r.t. the $\{ \partial_{m\alpha} \mid m \geqslant 0, 1 \leqslant \alpha \leqslant r\}$ and the first set of equations in (\ref{Lax2combAKNS}) imply that, when you translate the $\{W_{\beta}\}$ back to ${\rm gl}_{n}(R)[z, z^{-1})$ by $V_{\beta}(z)=W_{\beta}(\frac{1}{z})$, then these $\{V_{\beta}\}$ solve
 the strict $({\rm sl}_{n}(\mathbb{C}), \mathfrak{t})$-hierarchy w.r.t. the $\{ \partial_{m\beta} \mid m< 0,1 \leqslant \beta \leqslant r\}$.
Therefore we call the equations (\ref{Lax1combAKNS}) and (\ref{Lax2combAKNS}) the {\it Lax equations of the combined $({\rm sl}_{n}(\mathbb{C}), \mathfrak{t})$-hierarchy} and the deformation $(\{U_{\alpha}\}, \{W_{\beta}\})$ satisfying these equations a {\it solution of the combined $({\rm sl}_{n}(\mathbb{C}), \mathfrak{t})$-hierarchy}. Note that the trivial perturbation $(\{E_{\alpha}\}, \{ E_{\beta}z^{-1}\})$ solves this system as in the unperturbed situation all elements of the basis of $C$ commute and moreover, are constants for all the derivations $\{ \partial_{m\alpha}\}$. We refer to it as the {\it trivial} solution.

Also the system of Lax equations (\ref{Lax1combAKNS}) and (\ref{Lax2combAKNS}) is compatible, for there holds
\begin{proposition}
\label{P3.1}
Let the deformation $(\{U_{\alpha}\}, \{W_{\beta}\})$ be a solution of the combined $({\rm sl}_{n}(\mathbb{C}), \mathfrak{t})$-hierarchy and consider the projections $\{ B_{m\alpha}:=\pi_{\geqslant 0}(U_{\alpha}z^{m}) \}$ 
resp. $\{ C_{m\beta}:=\pi_{< 0}(W_{\beta}z^{m+1}) \}$ 
that occur in the Lax equations of the combined hierarchy.
Now these projections satisfy the following zero curvature relations: for all $\{ \alpha_{i}\}$ and $\{\beta_{j}\}$ in $[1,r] \cap \mathbb{Z}$
\begin{align}
\label{ZC1AKNS}
&\partial_{m_{1}\beta_{1}}(B_{m_{2}\alpha_{2}})-\partial_{m_{2}\alpha_{2}}(C_{m_{1}\beta_{1}})-[C_{m_{1}\beta_{1}},B_{m_{2}\alpha_{2}}]=0 \text{ if $m_{1} <0$, $m_{2} \geqslant 0$,}\\
\label{ZC2AKNS}
&\partial_{m_{1}\alpha_{1}}(B_{m_{2}\alpha_{2}})-\partial_{m_{2}\alpha_{2}}(B_{m_{1}\alpha_{1}})-[B_{m_{1}\alpha_{1}},B_{m_{2}\alpha_{2}}]=0 \text{ if $m_{1} \geqslant 0$, $m_{2} \geqslant 0$,}\\
\label{ZC3AKNS}
&\partial_{m_{1}\beta_{1}}(C_{m_{2}\beta_{2}})-\partial_{m_{2}\beta_{2}}(C_{m_{1}\beta_{1}})-[C_{m_{1}\beta_{1}},C_{m_{2}\beta_{2}}]=0 \text{ if $m_{1} <0$, $m_{2} < 0$.}
\end{align}
\end{proposition}
\begin{proof}
The relations (\ref{ZC2AKNS}) and (\ref{ZC3AKNS}) follow from Proposition \ref{P2.1}, so we merely have to proof the mixed relation (\ref{ZC1AKNS}).
The main idea of the proof is to show that the left hand side of the equation in (\ref{ZC1AKNS}) belongs to both $\pi_{\geqslant 0}({\rm sl}_{n}(R)[z, z^{-1}))$ and $\pi_{<0}({\rm sl}_{n}(R)[z, z^{-1}))$ and therefore has to be equal to zero.

Since the powers of $z$ are central and the second set of equations in (\ref{Lax2combAKNS}) holds for the $\{U_{\alpha}\}$, we know that we have for all $m_{2}\geqslant 0$ and all $m_{1} <0$
$$
\partial_{m_{1}\beta_{1}}(U_{\alpha_{2}}z^{m_{2}})=[\pi_{ < 0}(W_{\beta_{1}}z^{m_{1}+1}),U_{\alpha_{2}}z^{m_{2}}]=[C_{m_{1}\beta_{1}}, U_{\alpha_{2}}z^{m_{2}}].
$$
Combining this with the substitution $B_{m_{2}\alpha_{2}}=U_{\alpha_{2}}z^{m_{2}}-\pi_{<0}(U_{\alpha_{2}}z^{m_{2}})$ we get that
$$
\partial_{m_{1}\beta_{1}}(B_{m_{2}\alpha_{2}})-[C_{m_{1}\beta_{1}},B_{m_{2}\alpha_{2}}]=-\partial_{m_{1}\beta_{1}}(\pi_{<0}(U_{\alpha_{2}}z^{m_{2}}))+[C_{m_{1}\beta_{1}},\pi_{<0}(U_{\alpha_{2}}z^{m_{2}})]
$$
and the right hand side of this expression clearly belongs to $\pi_{<0}({\rm sl}_{n}(R)[z, z^{-1}))$. Now $\partial_{m_{2}\alpha_{2}}(C_{m_{1}\beta_{1}})$ also belongs to this Lie subalgebra, so we see that the whole left hand side of (\ref{ZC1AKNS}) lies inside $\pi_{<0}({\rm sl}_{n}(R)[z, z^{-1}))$.

To get the other inclusion, we use the second set of Lax equations for the $\{W_{\beta}\}$ in (\ref{Lax1combAKNS}). For the same reason as above, we get then that for all $m_{1} <0$ and 
all $m_{2}\geqslant 0$ there holds
$$
\partial_{m_{2}\alpha_{2}}(W_{\beta_{1}}z^{m_{1}+1})=[(U_{\alpha_{2}}z^{m_{2}})_{ \geqslant 0},W_{\beta_{1}}z^{m_{1}+1}]=[B_{m_{2}\alpha_{2}},W_{\beta_{1}}z^{m_{1}+1}].
$$
Again we combine this expression with the substitution $C_{m_{1}\beta_{1}}=W_{\beta_{1}}z^{m_{1}+1}-\pi_{\geqslant 0}(W_{\beta_{1}}z^{m_{1}+1})$ and obtain that
$$
-\partial_{m_{2}\alpha_{2}}(C_{m_{1}\beta_{1}})-[C_{m_{1}\beta_{1}},B_{m_{2}\alpha_{2}}]=\partial_{m_{2}\alpha_{2}}(\pi_{\geqslant 0}(W_{\beta_{1}}z^{m_{1}+1}))+ [\pi_{\geqslant 0}(W_{\beta_{1}}z^{m_{1}+1}),B_{m_{2}\alpha_{2}}]
$$
The right hand side of this expression clearly belongs to the Lie algebra $\pi_{\geqslant 0}({\rm sl}_{n}(R)[z, z^{-1}))$. The same is true for the term $\partial_{m_{1}\beta_{1}}(B_{m_{2}\alpha_{2}})$ and that proves the other inclusion.
\end{proof}

The reverse statement also holds:

\begin{proposition} 
\label{P3.2}
Suppose we have a deformation $\{U_{\alpha} \}$ of the type (\ref{decoU}) and a deformation $\{W_{\beta} \}$ of the form (\ref{defS}). Assume that the two sets of projections $\{ B_{m\alpha}:=\pi_{\geqslant 0}(U_{\alpha}z^{m}) \}$
and 
$\{ C_{m\beta}:=\pi_{< 0}(W_{\beta}z^{m+1}) \}$
satisfy the zero curvature relations (\ref{ZC1AKNS}), (\ref{ZC2AKNS}) and (\ref{ZC3AKNS}). Then the deformation $(\{U_{\alpha} \},\{W_{\beta} \})$ of the initial basis $\{E_{\alpha}z^{m}\}$ of $C$ is a solution of the combined $({\rm sl}_{n}(\mathbb{C}), \mathfrak{t})$-hierarchy.
\end{proposition}
\begin{proof}
It has been shown in Proposition \ref{P2.2} that the set of zero curvature relations (\ref{ZC2AKNS}) suffice to prove the first set of Lax equations in (\ref{Lax1combAKNS}).
Also the first set of Lax equations in (\ref{Lax2combAKNS}) follows from the zero curvature relations (\ref{ZC3AKNS}). So, assume first that one of the remaining Lax equations for one of the $\{U_{\alpha}\}$ does not hold, then there is a $\ell_{1}<0$ such that
$$
\partial_{\ell_{1}\beta}(U_{\alpha})-[(W_{\beta}z^{\ell_{1}+1})_{ < 0},U_{\alpha}]=\partial_{\ell_{1}}(U_{\alpha})-[C_{\ell_{1}\beta},U_{\alpha}]= \sum_{k \leqslant k_{2}}A_{k}z^{k}, \text{ with } A_{k_{2}} \neq 0.
$$
This implies for all $\ell \geqslant 0$ that
\begin{equation}
\label{Qm2}
\partial_{\ell_{1}\beta}(U_{\alpha}z^{\ell})-[C_{\ell_{1}\beta},U_{\alpha}z^{\ell}]= \sum_{k \leqslant k_{2}}A_{k}z^{k+\ell}
\end{equation}
and by letting $\ell$ go to infinity, you see that nonzero terms with unlimited high powers of $z$ occur in the expression (\ref{Qm2}). On the other hand, we can split the expression and substitute the identity (\ref{ZC1AKNS}) for $m_{1}=\ell_{1}$ and $m_{2}=\ell$, yielding
\begin{align*}
\partial_{\ell_{1}\beta}(U_{\alpha}z^{\ell})-[C_{\ell_{1}\beta},U_{\alpha}z^{\ell}]&= 
\partial_{\ell_{1}\beta}(B_{\ell \alpha})-[C_{\ell_{1}\beta},B_{\ell \alpha}]
+\partial_{\ell_{1}\beta}(\pi_{<0}(U_{\alpha}z^{\ell}))\\
&\;\;\;-[C_{\ell_{1}\beta},\pi_{<0}(U_{\alpha}z^{\ell})]\\
&=\partial_{\ell \alpha}(C_{\ell_{1}\beta})+\partial_{\ell_{1}\beta}(\pi_{<0}(U_{\alpha}z^{\ell}))-[C_{\ell_{1}\beta},\pi_{<0}(U_{\alpha}z^{\ell})]
\end{align*}
and this last expression has only negative powers of $z$. This contradicts the unlimited growth of these powers, so all the Lax equations for the $\{U_{\alpha}\}$ have to hold. So, one can only have one of the remaining Lax equations of one of the $\{W_{\beta}\}$ to be wrong. Suppose, there is a $s_{1}\geqslant 0$ such that we have
$$
\partial_{s_{1}\alpha}(W_{\beta})-[(Qz^{s_{1}})_{ \geqslant 0},W_{\beta}]=\partial_{s_{1}\alpha}(W_{\beta})-[B_{s_{1}\alpha},W_{\beta}]=\sum_{k \geqslant k_{1}}D_{k}z^{k}, \text{ with } D_{k_{1}} \neq 0.
$$
Then we get similarly for all $s < 0$ that
\begin{equation}
\label{Ss1}
\partial_{s_{1}\alpha}(W_{\beta}z^{s+1})-[B_{s_{1}\alpha},W_{\beta}z^{s+1}]= \sum_{k \geqslant k_{1}}D_{k}z^{k+s+1}
\end{equation}
and by letting $s$ go to minus infinity, one sees that there is no lower bound for the powers of $z$ in expression (\ref{Ss1}) with nonzero coefficients. Again we split the expression and substitute relation (\ref{ZC1AKNS}) for $m_{1}=s$ and $m_{2}=s_{1}$. This results in
\begin{align*}
\partial_{s_{1}\alpha}(W_{\beta}z^{s+1})-[B_{s_{1}\alpha},W_{\beta}z^{s+1}]= 
\partial_{s_{1}\alpha}(C_{s})-[B_{s_{1}\alpha},C_{s}]
+\partial_{s_{1}\alpha}(\pi_{\geqslant 0}(W_{\beta}z^{s+1}))\\-[B_{s_{1}\alpha},\pi_{\geqslant 0}(W_{\beta}z^{s+1})]
=\partial_{s\beta}(B_{s_{1}\alpha})+\partial_{s_{1}\alpha}(\pi_{\geqslant 0}(W_{\beta}z^{s+1}))-[B_{s_{1}\alpha},\pi_{\geqslant 0}(W_{\beta}z^{s+1})]
\end{align*}
and now this last expression contains only positive powers of $z$. Again a contradiction and thus also the Lax equations for all the $\{W_{\beta}\}$ hold. This proves this Proposition.
\end{proof}

\section{The linearization of the combined $({\rm sl}_{n}(\mathbb{C}), \mathfrak{t})$-hierarchy}
\label{linhier}

The zero curvature form of the combined $({\rm sl}_{n}(\mathbb{C}), \mathfrak{t})$-hierarchy points at the possible existence of a linear system of which the zero curvature equations form the compatibility conditions. In \cite{H2016} we produced such linearizations for both the AKNS-hierarchy and its strict version. We adept those to the setting here and the presence of the additional variables. 

So we start with a deformation $(\{U_{\alpha} \},\{W_{\beta} \})$ of the initial basis $\{E_{\alpha}z^{m}\}$ of $C$ that is a potential solution to the combined $({\rm sl}_{n}(\mathbb{C}), \mathfrak{t})$-hierarchy
As in the foregoing Section we associate with such a pair two sets of projections the $\{ B_{m\alpha}:=\pi_{\geqslant 0}(U_{\alpha}z^{m})  \}$ and the 
$\{ C_{m\beta}:=\pi_{< 0}(W_{\beta}z^{m+1})  \}$.  
Then the {\it linearization of the combined $({\rm sl}_{n}(\mathbb{C}), \mathfrak{t})$-hierarchy} consists of the system
\begin{align}
\label{Lin1combAKNS}
&U_{\alpha} \psi= \psi E_{\alpha}, \partial_{m \alpha}(\psi)=B_{m \alpha}\psi, \text{ for all $m \geqslant 0, \alpha, \alpha \in [1,r] \cap \mathbb{Z},$ and } \\ \notag
&\partial_{m\beta}(\psi)=C_{m\beta}\psi, \text{ for all $m < 0$}, \beta \in [1,r] \cap \mathbb{Z},\\
\label{Lin2combAKNS}
&W_{\beta} \varphi= \varphi E_{\beta}z^{-1}, \partial_{m\beta}(\varphi)=C_{m\beta}\varphi, \text{ for all $m <0 , \beta \in [1,r] \cap \mathbb{Z},$ and }\\ \notag 
&\partial_{m\alpha}(\varphi)=B_{m\alpha}\varphi, \text{ for all $m \geqslant 0, \alpha \in [1,r] \cap \mathbb{Z}$}.
\end{align}
Without specifying $\psi$ and $\varphi$, we first show what is needed to get from (\ref{Lin1combAKNS}) and (\ref{Lin2combAKNS}) the Lax
equations for $\{U_{\alpha} \}$ and the $\{W_{\beta} \}$. In fact, we give the manipulations for the $\{U_{\alpha} \}$, those for the $\{W_{\beta} \}$ are similar. First we apply $\partial_{m \alpha_{1}}, m \geqslant 0, \alpha_{1} \in [1,r] \cap \mathbb{Z},$ to the first equation in (\ref{Lin1combAKNS}) and use the first two equations in the sequel
\begin{align}
\notag
\partial_{m\alpha_{1}}(U_{\alpha_{2}}\psi -\psi E_{\alpha_{2}})&=\partial_{m\alpha_{1}}(U_{\alpha_{2}} )\psi +U_{\alpha_{2}}\partial_{m\alpha_{1}}(\psi)-\partial_{m\alpha_{1}}(\psi)E_{\alpha_{2}}=0\\ \notag
&=\partial_{m\alpha_{1}}(U_{\alpha_{2}} )\psi+U_{\alpha_{2}}B_{m\alpha_{1}} \psi-B_{m\alpha_{1}} \psi E_{\alpha_{2}}\\ \label{Lincomb1Lax}
&=\left\{ \partial_{m\alpha_{1}}(U_{\alpha_{2}})-[B_{m\alpha_{1}},U_{\alpha_{2}}] \right\} \psi=0.
\end{align}
Now we carry out the same computation for $\partial_{m\beta}, m < 0,\beta \in [1,r] \cap \mathbb{Z},$ make use of the first and third equation in (\ref{Lin1combAKNS}) and obtain:
\begin{align}
\notag
\partial_{m\beta}(U_{\alpha}\psi -\psi E_{\alpha})&=\partial_{m\beta}(U_{\alpha} )\psi +U_{\alpha}\partial_{m\beta}(\psi)-\partial_{m\beta}(\psi)E_{\alpha}=0\\ \notag
&=\partial_{m\beta}(U_{\alpha} )\psi+U_{\alpha}C_{m\beta} \psi-C_{m\beta} \psi E_{\alpha}\\ \label{Lincomb2Lax}
&=\left\{ \partial_{m\beta}(U_{\alpha})-[C_{m\beta},U_{\alpha}] \right\} \psi=0.
\end{align}
If we can scratch $\psi$ from both equations (\ref{Lincomb1Lax}) and (\ref{Lincomb2Lax}), then we obtain the desired Lax equations for each $U_{\alpha}$. Summarizing the manipulations carried out, we need, first of all, a left action of elements like $U_{\alpha}$, $B_{m\alpha}$ and $C_{m\beta}$. Next there should be a right action of $E_{\alpha}$ and an appropriate left action of all the $\partial_{m\alpha}, m \in \mathbb{Z},$ that obeys a Leibnitz rule w.r.t. the action of the elements from ${\rm sl}_{n}(R)[z, z^{-1})$ and finally the scratch procedure. This can all be realized for suitable $\psi$ in an appropriate ${\rm gl}_{n}(R)[z, z^{-1})$-module. Similarly, one can deduce the Lax equations for the $\{W_{\beta} \}$ from (\ref{Lin2combAKNS}) if $\varphi$ is a suitable vector in a certain ${\rm gl}_{n}(R)[z^{-1}, z)$-module.

To get an idea of these modules, we first have a look at the linearization for the trivial solutions $U_{\alpha}=E_{\alpha}$ and $W_{\beta}=E_{\beta}z^{-1}$. Then 
the projections are $B_{m\alpha}=E_{\alpha}z^{m}, m \geqslant 0, \alpha \in [1,r] \cap \mathbb{Z},$ and $C_{m\beta}=E_{\beta}z^{m}, m<0, \beta \in [1,r] \cap \mathbb{Z},$ and the equations of the linearization become
\begin{align}
\label{tsol1combAKNS}
&E_{\alpha} \psi_{0}= \psi_{0} E_{\alpha}, \partial_{m\alpha}(\psi_{0})=E_{\alpha}z^{m}\psi_{0}, \text{, $m \geqslant 0,\alpha \in [1,r] \cap \mathbb{Z}$, and }\\ \notag 
&\partial_{m\beta}(\psi_{0})=E_{\beta}z^{m}\psi_{0}, \text{ $m < 0$}, \beta \in [1,r] \cap \mathbb{Z},\\
\label{tsol2combAKNS}
&E_{\alpha}z^{-1} \varphi_{0}= \varphi_{0} E_{\alpha}z^{-1}, \partial_{m\beta}(\varphi)=E_{\beta}z^{m}\varphi, \text{, $m <0, \beta \in [1,r] \cap \mathbb{Z} $ and } \\ \notag
&\partial_{m\alpha}(\varphi)=E_{\alpha}z^{m}\varphi, \text{ $m \geqslant 0, \alpha \in [1,r] \cap \mathbb{Z}$}.
\end{align}
Assuming that each derivation $\partial_{m\alpha}$ equals $\frac{\partial}{\partial t_{m \alpha}}$ and writing $t$ as a short hand notation for all the flow parameters $\{t_{m\alpha} \mid m \in \mathbb{Z}, \alpha \in [1,r] \cap \mathbb{Z}\}$, one arrives for (\ref{tsol1combAKNS}) and (\ref{tsol2combAKNS}) at the solution $(\psi_{0},\varphi_{0})$ 
$$ 
\psi_{0}=\psi_{0}(t,z)=\exp(\sum_{m \in \mathbb{Z}} \sum_{\alpha =1}^{r}t_{m\alpha}E_{\alpha}z^{m})=\varphi_{0}(t,z)=\varphi_{0}.
$$
General $\psi$ should be ${\rm gl}_{n}(R)[z, z^{-1})$-perturbations of $\psi_{0}$, i.e. they should belong to 
\begin{equation}
\label{Mgeq0}
\mathcal{M}_{\geqslant 0}=\left\{ \{g(z)\}\psi_{0}
\mid g(z)=\sum_{i=-\infty}^{N}g_{i}z^{i} \in {\rm gl}_{n}(R)[z, z^{-1}) \right\}
\end{equation}
and general $\varphi $ should be ${\rm gl}_{n}(R)[z^{-1}, z)$-perturbations of $\varphi_{0}$, i.e. they should belong to 
 \begin{equation}
\label{M<0}
\mathcal{M}_{<0}=\left\{ \{h(z)\} \varphi_{0}
\mid h(z)=\sum_{i=-N}^{\infty}h_{i}z^{i} \in {\rm gl}_{n}(R)[z^{-1},z) \right\},
\end{equation}
where the products $\{g(z)\}\psi_{0}$ and $\{h(z)\} \varphi_{0}$ should be seen as formal and both factors should be kept separate to avoid convergence issues.
On both $\mathcal{M}_{\geqslant 0}$ and $\mathcal{M}_{<0}$ one can define the required actions: for each $k_{1}(z) \in {\rm gl}_{n}(R)[z, z^{-1})$ and each $k_{2}(z) \in {\rm gl}_{n}(R)[z^{-1}, z)$ define
\[
k_{1}(z).\{g(z)\}\psi_{0}:=\{k_{1}(z)g(z)\}\psi_{0} \text{ resp. }k_{2}(z).\{h(z)\} \varphi_{0}:=\{k_{2}(z)h(z)\} \varphi_{0}.
\]
The right hand action of each $E_{\alpha}$ on $\mathcal{M}_{\geqslant 0}$ resp. $E_{\alpha}z^{-1}$ on $\mathcal{M}_{< 0}$ we define by
\[
\{g(z)\}\psi_{0}E_{\alpha}:=\{g(z)E_{\alpha} \}\psi_{0} \text{ resp. }\{h(z)\} \varphi_{0}E_{\alpha}z^{-1}:=\{h(z)E_{\alpha}z^{-1}\} \varphi_{0}
\]
and the action of each $\partial_{m\alpha}$ by
\begin{align*}
\partial_{m\alpha}(\{g(z)\}\psi_{0})=\left\{ \sum_{i=-\infty}^{N}\partial_{m\alpha}(g_{i})z^{i} +\left\{ \sum_{i=-\infty}^{N}g_{i}E_{\alpha}z^{i+m} \right\} \right\}\psi_{0},\\
\partial_{m\alpha}(\{h(z)\}\varphi_{0})=\left\{ \sum_{i=-N}^{\infty}\partial_{m\alpha}(h_{i})z^{i} +\left\{ \sum_{i=-N}^{\infty}h_{i}E_{\alpha}z^{i+m} \right\} \right\}\varphi_{0}.
\end{align*}
Analogous to the terminology in the scalar case, see \cite{DJKM83}, we call the elements of $\mathcal{M}_{\geqslant 0}$ {\it oscillating matrices at infinity} and those of $\mathcal{M}_{<0}$ {\it oscillating matrices at zero}.
Note that $\mathcal{M}_{\geqslant 0}$ is a free ${\rm gl}_{n}(R)[z, z^{-1})$-module and $\mathcal{M}_{<0}$ a free ${\rm gl}_{n}(R)[z^{-1},z)$-modules with respective generators $\psi_{0}$ and $\varphi_{0}$, because for 
each $k_{1}(z) \in {\rm gl}_{n}(R)[z, z^{-1})$ and $k_{1}(z) \in {\rm gl}_{n}(R)[ z^{-1},z)$ we have
\[
k_{1}(z).\psi_{0}=k_{1}(z).\{1\}\psi_{0}=\{k_{1}(z)\}\psi_{0} \text{ resp. }k_{2}(z).\varphi_{0}=k_{2}(z).\{1\}\psi_{0}=\{k_{2}(z)\}\varphi_{0}. 
\]
Hence, in order to be able to perform legally the scratching of both vectors $\psi=\{k_{1}(z)\}\psi_{0}$ and $\varphi=\{k_{2}(z)\}\varphi_{0}$, it is enough to find oscillating matrices such that 
$k_{1}(z)$ is invertible in ${\rm gl}_{n}(R)[z, z^{-1})$ and likewise $k_{2}(z)$ in ${\rm gl}_{n}(R)[ z^{-1},z)$. We will now introduce a collection of such elements that will occur at the construction of solutions of the hierarchy.

Recall from Section 3 that we may assume that $\mathfrak{t}$ is realized in the upper triangular matrices. For each $l=(l_{i}) \in \mathbb{Z}^{n},$ we define the $n \times n$-matrix $\delta(l)$ by
 $$
 \delta(l)=\left(
\begin{matrix}
z^{l_{1}}& 0&0 \\
0&\ddots &0 \\
0& 0& z^{l_{n}}
\end{matrix}
\right) \in  {\rm gl}_{n}(R)[z,z^{-1}].
 $$
The collection of all these matrices forms a group $\Delta$ and we consider its subgroup $\Delta(\mathfrak{t})$
of all matrices 
$$
\Delta(\mathfrak{t})=\left\{  \delta(l) \in \Delta  \mid [ \delta(l),E_{\alpha}]=0 \text{ for all } \alpha, 1 \leqslant \alpha \leqslant r \right\}.
$$

For each $\delta(l) \in \Delta(\mathfrak{t})$, an element $\psi \in \mathcal{M}_{\geqslant 0}$ is called an {\it oscillating matrix at infinity of type} $\delta(l), $ if it has the form 
\begin{equation}
\label{omidm}
\psi=\{ k_{1}(z) \delta(l) \}\psi_{0}, \text{ with } k_{1}(z) \in G_{<0},
\end{equation}
and is an example of a generator of $\mathcal{M}_{\geqslant 0}$. Similarly, an element $\varphi \in \mathcal{M}_{< 0}$ is called an {\it oscillating matrix at zero of type} $\delta(l), \text{ with }\delta(l) \in \Delta(\mathfrak{t}),$ if it has the form 
\begin{equation}
\label{omzdm}
\varphi=\{ k_{2}(z) \delta(l) \}\varphi_{0}, \text{ with } k_{2}(z) \in G_{\geqslant 0},
\end{equation}
and such a $\varphi$ generates $\mathcal{M}_{< 0}$. Hence, for any pair $(\psi,\varphi) \in \mathcal{M}_{\geqslant 0} \times \mathcal{M}_{< 0}$ with $\psi$ of the form (\ref{omidm}) and $\varphi$ of the form (\ref{omzdm}) the scratching procedure is valid. 

Now let the deformation $(\{U_{\alpha} \},\{W_{\beta} \})$ be a potential solution of the combined $({\rm sl}_{n}(\mathbb{C}), \mathfrak{t})$-hierarchy and let $(\psi,\varphi)$ be a pair in $ \mathcal{M}_{\geqslant 0} \times \mathcal{M}_{< 0}$ with $\psi$ of the form (\ref{omidm}) and $\varphi$ of the form (\ref{omzdm}), for which the linearization equations 
(\ref{Lincomb1Lax}) and (\ref{Lincomb2Lax}) hold. Then all the manipulations necessary to get the Lax equations (\ref{Lax1combAKNS}) and 
(\ref{Lax2combAKNS}), are valid. Hence, the set $(\{U_{\alpha} \},\{W_{\beta} \})$ is a solution of the combined $({\rm sl}_{n}(\mathbb{C}), \mathfrak{t})$-hierarchy, and we call the pair $(\psi,\varphi)$ a set of {\it wave matrices of the combined $({\rm sl}_{n}(\mathbb{C}), \mathfrak{t})$-hierarchy of type} $\delta(l)$. In particular, the pair $(\psi,\varphi)$ totally determines the solution $(\{U_{\alpha} \},\{W_{\beta} \})$, for the first equations in (\ref{Lincomb1Lax}) and (\ref{Lincomb2Lax}), imply respectively
\begin{align*} 
&U_{\alpha} k_{1}(z) \delta(l)=k_{1}(z)E_{\alpha}\delta(l) \Rightarrow U_{\alpha}=k_{1}(z)E_{\alpha}k_{1}(z)^{-1},\\
&W_{\beta} k_{2}(z) \delta(l)=k_{2}(z)E_{\beta}z^{-1}\delta(l) \Rightarrow W_{\beta}=k_{2}(z)E_{\beta}z^{-1}k_{2}(z)^{-1}.
\end{align*}
There is a milder condition that pairs of oscillating matrices of a certain type have to satisfy, in order to become a set of wave matrices of the combined $({\rm sl}_{n}(\mathbb{C}), \mathfrak{t})$-hierarchy.
\begin{proposition}
\label{P4.1}
Let $\psi=\{ k_{1}(z) \delta(l) \}\psi_{0}$ be an oscillating matrix of type $\delta(l)$ in $\mathcal{M}_{\geqslant 0}$ and let $\varphi=\{ k_{2}(z) \delta(l) \}\varphi_{0}$ be such a matrix in $\mathcal{M}_{<0}$. Denote the corresponding potential solution of the combined $({\rm sl}_{n}(\mathbb{C}), \mathfrak{t})$-hierarchy by
$$
U_{\alpha}:=k_{1}(z)E_{\alpha}k_{1}(z)^{-1}, \text{ resp.} W_{\beta}=k_{2}(z)E_{\beta}z^{-1}k_{2}(z)^{-1}.
$$
If there exists for each $m \geqslant 0,  1 \leqslant \alpha \leqslant r,$ an element $M_{m\alpha} \in \pi_{\geqslant 0}({\rm gl}_{n}(R)[z, z^{-1}))$ such that 
\[
\partial_{m\alpha}(\psi)=M_{m\alpha}\psi \text{ and } \partial_{m\alpha}(\varphi)=M_{m\alpha}\varphi
\]
and, moreover, for all $m<0, 1 \leqslant \beta \leqslant r,$ there exists an $N_{m\beta} \in \pi_{> 0}({\rm gl}_{n}(R)[z^{-1}, z))$ such that 
\[
\partial_{m\beta}(\psi)=N_{m\beta}\psi \text{ and }\partial_{m\beta}(\varphi)=N_{m\beta}\varphi.
\]
Then each $M_{m\alpha}=\pi_{\geqslant 0}(E_{\alpha}z^{m}), $ and each $N_{m\beta}=\pi_{<0}(W_{\beta}z^{m+1})$ and the pair $(\psi,\varphi)$ is a set of wave matrices for the combined $({\rm sl}_{n}(\mathbb{C}), \mathfrak{t})$-hierarchy of type $\delta(l)$.
\end{proposition}
\begin{proof} 
For $m \geqslant 0$ we use the fact that $\mathcal{M}_{\geqslant 0}$ is a free ${\rm gl}_{n}(R)[z, z^{-1})$-module
with generator $\psi_{0}$. This property enables us to translate the equation $\partial_{m\alpha}(\psi)=M_{m\alpha}\psi$ into an equation in ${\rm gl}_{n}(R)[z, z^{-1})$:
$$
\partial_{m\alpha}(k_{1}(z))+k_{1}(z)E_{\alpha}z^{m}=M_{m\alpha} k_{1}(z) \Rightarrow \partial_{m\alpha}(k_{1}(z))k_{1}(z)^{-1} +U_{\alpha}z^{m}=M_{m\alpha}.
$$
Projecting this onto $\pi_{\geqslant 0}({\rm gl}_{n}(R)[z, z^{-1}))$ 
yields the formula $M_{m\alpha}=\pi_{\geqslant 0}(U_{\alpha}z^{m}).$
If $m<0$, then one uses the property that $\mathcal{M}_{<0}$ is a free ${\rm gl}_{n}(R)[z^{-1},z)$-module
with generator $\varphi_{0}$. Translating the relations $\partial_{m\beta}(\varphi)=N_{m\beta}\varphi$ into equations in ${\rm gl}_{n}(R)[z^{-1},z)$ yields:
$$
\partial_{m\beta}(k_{2}(z))+k_{2}(z)E_{\alpha}z^{m}=N_{m\beta}k_{2}(z) \Rightarrow \partial_{m\beta}(k_{2}(z))k_{2}(z)^{-1} +W_{\beta}z^{m+1}=N_{m\beta}
$$
Projecting the right hand side on $\pi_{<0}({\rm gl}_{2}(R)[z^{-1}, z)$ gives us the identity we are looking for: $\pi_{<0}(W_{\beta}z^{m+1})=N_{m\beta}$.
\end{proof}

\begin{remark}
\label{R4.1}
This concludes the presentation of the algebraic framework of the linearization of the combined $({\rm sl}_{n}(\mathbb{C}), \mathfrak{t})$-hierarchy. In the next section, we present an analytic context where we can construct sets of wave matrices of this hierarchy in which the products are no longer formal, but real. 
\end{remark}

\section{A construction of solutions of the hierarchy}
\label{constr}

In this section we will show how to construct a wide class of solutions of the combined $({\rm sl}_{n}(\mathbb{C}), \mathfrak{t})$-hierarchy. This is done in the style of \cite{SW85} and \cite{PS86}. We first describe the group of loops we will work with. 
For each $0 < r<1$, let $A_{r}$ be the annulus in the complex plane given by
$$A_{r}=\{ z\mid z \in \mathbb{C}, r \leqslant |z| \leqslant \dfrac{1}{r} \}.$$  
We denote the collection of holomorphic maps from some annulus $A_{r}$ into $\GL_{n}(\mathbb{C})$ by  $L_{an}\GL_{n}(\mathbb{C})$. 
It is a group w.r.t. point wise multiplication and contains in a natural way $\GL_{n}(\mathbb{C})$ as a subgroup as the collection of constant maps into $\GL_{n}(\mathbb{C})$. Other examples of elements in $L_{an}\GL_{n}(\mathbb{C})$  are the elements of $\Delta$. However, $L_{an}\GL_{n}(\mathbb{C})$ is more than just a group, it is an infinite dimensional Lie group.
Its manifold structure comes from
its Lie algebra $L_{an}\gl_{n}(\mathbb{C})$ consisting of all holomorphic maps 
$\gamma :U \to \gl_{n}(\mathbb{C})$, where $U$ is an open neighborhood of some annulus $A_{r}, 0 < r<1.$ Since $\gl_{n}(\mathbb{C})$ is a Lie algebra, the space $L_{an}\gl_{n}(\mathbb{C})$ becomes a Lie algebra w.r.t. the point wise commutator. Topologically, the space $L_{an}\gl_{n}(\mathbb{C})$ is the direct limit of all the spaces $L_{an,r}\gl_{n}(\mathbb{C})$, where this last space consists of all $\gamma$ corresponding to the fixed annulus $A_{r}$. One gives each $L_{an,r}\gl_{n}(\mathbb{C})$ the topology of uniform convergence and with that topology it becomes a Banach space. In this way, $L_{an}\gl_{n}(\mathbb{C})$ becomes a Fr\'echet space. The point wise exponential map defines a local diffeomorphism around zero in $L_{an}\gl_{n}(\mathbb{C})$, see e.g. \cite{Hamilton82}.

Now each loop  $\ell \in L_{an}\gl_{n}(\mathbb{C})$ 
possesses an expansion in a Fourier series
\begin{equation*}
\ell= \sum_{k=-\infty}^{\infty}
\ell_{k}
z^{k}, \text{ with each }\ell_{k} \in \gl_{n}(\mathbb{C}),
\end{equation*}
that converges absolutely on the annulus it is defined:
\[
\sum_{k=-\infty}^{\infty} ||\ell_{k}||r^{-|k|} < \infty.
\]
This Fourier expansion is used to make the relevant decomposition of the Lie algebra $L_{an}\gl_{n}(\mathbb{C})$.  
Namely, consider the subspaces 
\begin{align*}
L_{an}\gl_{n}(\mathbb{C})_{\geqslant 0}:=\{ \ell \mid \ell \in L_{an}\gl_{n}(\mathbb{C}),\ell =\sum_{k=0}^{\infty}\ell_{k}z^{k}\}\\
L_{an}\gl_{n}(\mathbb{C})_{< 0}:=\{ \ell \mid \ell \in L_{an}\gl_{n}(\mathbb{C}),\ell =\sum_{k=-\infty}^{-1}\ell_{k}z^{k}\}
\end{align*}
Both are Lie subalgebras of $L_{an}\gl_{n}(\mathbb{C})$ and their direct sum equals the whole Lie algebra. The first Lie algebra consists of the elements in $L_{an}\gl_{n}(\mathbb{C})$ that extend to holomorphic maps defined on a disk around the origin of the form 
$$\{ z \in \mathbb{C} \mid |z| \leqslant \dfrac{1}{r} \},0< r <1,$$ 
and the second Lie algebra corresponds to the maps in $L_{an}\gl_{n}(\mathbb{C})$ that have a holomorphic extension towards a disk around infinity of the form 
$$\{z \in \mathbb{P}^{1}(\mathbb{C} ) \mid |z| \geqslant r \},0< r <1,$$ and that, moreover, are zero at infinity.
To each of the two Lie subalgebras belongs a subgroup of $L_{an}\GL_{n}(\mathbb{C})$. The point wise exponential map applied to elements of $L_{an}\gl_{n}(\mathbb{C})_{< 0}$ yields elements of 
\[
U_{-}=\{ \ell \mid \gamma \in L_{an}\gl_{n}(\mathbb{C}),\ell =\Id+\sum_{k=-\infty}^{-1}\ell_{k}z^{k}\}
\]
and the exponential map applied to elements of $L_{an}\gl_{n}(\mathbb{C})_{\geqslant 0}$ maps them into
\[
P_{+}=\{ \ell \mid \ell \in L_{an}\gl_{n}(\mathbb{C}),\ell =\ell_{0}+\sum_{k=1}^{\infty}\ell_{k}z^{k}, \text{ with }\ell_{0} \in \GL_{n}(\mathbb{C})\}.
\]
Both $U_{-}$ and $P_{+}$ are easily seen to be subgroups of $L_{an}\GL_{n}(\mathbb{C})$ and since the direct sum of their Lie algebras is $L_{an}\gl_{n}(\mathbb{C})$, their product 
\begin{equation}
\label{U-P+}
\Omega=U_{-}P_{+}
\end{equation}
is open in $L_{an}\GL_{n}(\mathbb{C})$ and is called, like in the finite dimensional case, the {\it big cell} w.r.t. $U_{-}$ and $P_{+}$.

The next subgroup of $L_{an}\SL_{n}(\mathbb{C})$ corresponds to the exponential factor in the linearization of the combined $({\rm sl}_{n}(\mathbb{C}), \mathfrak{t})$-hierarchy.
The commuting group relevant for this hierarchy is 
\[
\Gamma=\{ \gamma(t)=\exp(\sum_{m \in \mathbb{Z}} \sum_{\alpha=1}^{r}t_{m\alpha}E_{\alpha}z^{i}) \mid \gamma \in L_{an}\SL_{n}(\mathbb{C}) \}.
\]
The group $\Delta(\mathfrak{t})$ commutes with $\Gamma$ and contains 
the central subgroup 
$$
\Delta_{c}=\{ \delta^{k} \mid \delta=z \Id
, k \in \mathbb{Z} \}
$$ 
of $L_{an}\GL_{n}(\mathbb{C})$.

We have now all ingredients to describe the construction of the solutions 
to the combined $({\rm sl}_{n}(\mathbb{C}), \mathfrak{t})$-hierarchy.
Take inside the product $L_{an}\GL_{2}(\mathbb{C}) \times \Delta(\mathfrak{t})$ the collection $\mathcal{S}$ of pairs $(g,\delta(l))$ such that there exists a $\gamma(t), \gamma  \in \Gamma,$ satisfying 
\begin{equation}
\label{CC}
\delta(l) \gamma(t) g \gamma(t)^{-1}\delta(-l) \in \Omega=U_{-}P_{+}
\end{equation}
For each such a pair $(g,\delta(l))$, consider the collection $\Gamma(g,\delta(l))$ of all $\gamma(t) \in \Gamma$ satisfying the condition (\ref{CC}). This is an open non-empty subset of $\Gamma$. Let $R(g,\delta(l))$ be the algebra of analytic functions $\Gamma(g,\delta(l)) \to \mathbb{C}$. This is the algebra of functions $R$ that we associate with the point $(g,\delta(l)) \in \mathcal{S}$ and for the commuting derivations of $R(g,\delta(l))$ we choose the $$\partial_{m\alpha}:=\dfrac{\partial}{\partial t_{m \alpha}}, i \in \mathbb{Z}, 1 \leqslant \alpha \leqslant r.$$ By property (\ref{CC}), we have for all $\gamma(t) \in \Gamma(g,\delta(l))$
\begin{equation}
\label{decoU-P+}
\delta(l) \gamma(t) g \gamma(t)^{-1}\delta(-l)=u_{-}(g,\delta(l))(t)^{-1}p_{+}(g,\delta(l))(t), 
\end{equation}
$\text{with }u_{-}(g,\delta(l))(t) \in U_{-} \text{ and } p_{+}(g,\delta(l))(t) \in P_{+}.$
Then all the matrix coefficients in the Fourier expansions of the elements $u_{-}(g,\delta(l))$ and $p_{+}(g,\delta(l))$ belong to the algebra $R(g,\delta(l))$. From equation (\ref{decoU-P+}) one can build two oscillating matrices of type $\delta(l)$, one $\Psi_{g,\delta(l)} \in 
\mathcal{M}_{\geqslant 0}$ and the other one $\Phi_{g,\delta(l)} \in \mathcal{M}_{< 0}$. Define namely
\begin{align}
\label{wminf}
\Psi_{g,\delta(l)}(t):&=u_{-}(g,\delta(l))(t)\delta(l) \gamma(t),\\
\label{wmzero}
\Phi_{g,\delta(l)}(t):&=p_{+}(g,\delta(l))(t)\delta(l)\gamma(t),
\end{align}
and note that all the products between the different factors are well-defined. Due to relation (\ref{decoU-P+}), these two oscillating matrices of type $\delta(l)$ are related by
\begin{equation}
\label{relpsiphi}
\Psi_{g,\delta(l)}(t)=\Phi_{g,\delta(l)}(t) g^{-1}.
\end{equation}
From relation (\ref{decoU-P+}), one sees directly that for all $k \in \mathbb{Z}$ we have that if $(g,\delta(l)) \in \mathcal{S}$, then also each  $(g,\delta(l)\delta^{k}) \in \mathcal{S}$ and the sets of oscillating matrices relate according to
$$
\Psi_{g,\delta(l)\delta^{k}}=\Psi_{g,\delta(l)}\delta^{k} \text{ and } \Phi_{g,\delta(l)\delta^{k}}=\Phi_{g,\delta(l)}\delta^{k}.
$$
Now we want to show that each pair $(\Psi_{g,\delta(l)},\Phi_{g,\delta(l)})$ is a set of wave matrices of the combined $({\rm sl}_{n}(\mathbb{C}), \mathfrak{t})$-hierarchy, by using Proposition \ref{P4.1}. Thereto we compute for all $m \geqslant 0$ and $\alpha, 1 \leqslant \alpha \leqslant r$, in two ways $\partial_{m \alpha}(\Psi_{g,\delta(l)})$, once using (\ref{wminf}) and once using (\ref{wmzero}) and (\ref{relpsiphi}). This yields on one hand
\begin{align*}
&\partial_{m \alpha}(\Psi_{g,\delta(l)})=\{\partial_{m \alpha}(u_{-}(g,\delta(l)))+u_{-}(g,\delta(l))E_{\alpha}z^{m}\}\delta(l) \gamma=\\
&=\{\partial_{m \alpha}(u_{-}(g,\delta(l)))u_{-}(g,\delta(l))^{-1}+u_{-}(g,\delta(l))E_{\alpha}z^{m}u_{-}(g,\delta(l))^{-1}\}\Psi_{g,\delta(l)}
\end{align*}
and on the other
\begin{align*}
&\partial_{m \alpha}(\Psi_{g,\delta(l)})=\{\partial_{m \alpha}(p_{+}(g,\delta(l)))+p_{+}(g,\delta(l))E_{\alpha}z^{j}\}\delta(l) \gamma g^{-1}=\\
&=\{\partial_{m \alpha}(p_{+}(g,\delta(l)))p_{+}(g,\delta(l))^{-1}+p_{+}(g,\delta(l))E_{\alpha}z^{m}p_{+}(g,\delta(l))^{-1}\}\Psi_{g,\delta(l)}
\end{align*}
By comparing the two factors in front of $\Phi_{g,\delta(l)}$ in these expressions we see that
$$
M_{m\alpha}:=\partial_{m\alpha}(u_{-}(g,\delta(l)))u_{-}(g,\delta(l))^{-1}+u_{-}(g,\delta(l))E_{\alpha}z^{m}u_{-}(g,\delta(l))^{-1}
$$
belongs to $\pi_{\geqslant 0}({\rm gl}_{n}(R)[z, z^{-1}))$. Because of relation (\ref{relpsiphi}) we have that also for $\Phi_{g,\delta(l)}$, there 
holds for all $m \geqslant 0$ and $\alpha,1 \leqslant \alpha \leqslant r$
$$
\partial_{m\alpha}(\Phi_{g,\delta(l)})=M_{m\alpha}\Phi_{g,\delta(l)}.
$$
We proceed in a similar way with the computation of $\partial_{\beta}(\Phi_{g,\delta(l)})$ for all $m<0$ and $\beta, 1 \leqslant \beta \leqslant r$. Then we get the expressions
\begin{align*}
&\partial_{m \beta}(\Phi_{g,\delta(l)})=\{\partial_{m \beta}(p_{+}(g,\delta(l)))+p_{+}(g,\delta(l))E_{\beta}z^{m}\}\delta(l) \gamma g^{-1}=\\
&=\{\partial_{m \beta}(p_{+}(g,\delta(l)))p_{+}(g,\delta(l))^{-1}+p_{+}(g,\delta(l))E_{\beta}z^{m}p_{+}(g,\delta(l))^{-1}\}\Phi_{g,\delta(l)}
\end{align*}
and
\begin{align*}
&\partial_{m \beta}(\Phi_{g,\delta(l)})=\{\partial_{m \beta}(u_{-}(g,\delta(l)))+u_{-}(g,\delta(l))E_{\beta}z^{m}\}\delta(l) \gamma g=\\
&=\{\partial_{m \beta}(u_{-}(g,\delta(l)))u_{-}(g,\delta(l))^{-1}+u_{-}(g,\delta(l))E_{\beta}z^{m}u_{-}(g,\delta(l))^{-1}\}\Phi_{g,\delta(l)}
\end{align*}
Comparing the two factors in front of $\Phi_{g,\delta(l)}$ in these expressions yields that
$$
N_{m \beta}:=\partial_{m \beta}(p_{+}(g,\delta(l)))p_{+}(g,\delta(l))^{-1}+p_{+}(g,\delta(l))E_{\beta}z^{m}p_{+}(g,\delta(l))^{-1}
$$
belongs to $\pi_{< 0}({\rm gl}_{n}(R)[z^{-1}, z))$. Because of relation (\ref{relpsiphi}), also for $\Psi_{g,\delta(l)}$, there 
holds for all $j < 0$
$$
\partial_{m \beta}(\Psi_{g,\delta(l)})=N_{m \beta}\Psi_{g,\delta(l)}.
$$
So we have shown that all the conditions in Proposition \ref{P4.1} are satisfied, so that we may conclude

\begin{theorem}
\label{T5.1}
Consider the product space $\Pi:=L_{an}\GL_{2}(\mathbb{C}) \times \Delta(\mathfrak{t})$ and its subset $\mathcal{S}$ defined by (\ref{CC}). For each point $(g,\delta(l)) \in \mathcal{S}$, we define a pair of oscillating matrices $(\Psi_{g,\delta(l)},\Phi_{g,\delta(l)})$ in $\mathcal{M}_{\geqslant 0} \times \mathcal{M}_{< 0}$ by (\ref{wminf}) and (\ref{wmzero}). This pair is a set of wave matrices for the combined $({\rm sl}_{n}(\mathbb{C}), \mathfrak{t})$-hierarchy. In particular, the deformation $(\{U_{\alpha}(g,\delta(l)) \},\{W_{\beta}(g,\delta(l)) \})$ of the basis $\{ E_{\alpha}z^{m} \}$ of $C$ defined by
\begin{align*}
U_{\alpha}(g,\delta(l))&=u_{-}(g,\delta(l))E_{\alpha}u_{-}(g,\delta(l))^{-1} \text{ and }\\
W_{\beta}(g,\delta(l))&=p_{+}(g,\delta(l))E_{\beta}z^{-1}p_{+}(g,\delta(l))^{-1} ,
\end{align*}
forms a solution of the combined $({\rm sl}_{n}(\mathbb{C}), \mathfrak{t})$-hierarchy. This solution does not change if one replaces $\delta(l)$ by $\delta(l)\delta^{k}, k \in \mathbb{Z}.$
\end{theorem}

\begin{remark}
\label{R5.1}
If one is interested only in the solutions of the $({\rm sl}_{n}(\mathbb{C}), \mathfrak{t})$-hierarchy, one performs the construction in Theorem \ref{T5.1} with for $i<0$ all $t_{i \alpha}=0$ and the set $\{ U_{\alpha}\}$ is then a solution. Similarly, for solutions of the strict $({\rm sl}_{n}(\mathbb{C}), \mathfrak{t})$-hierarchy, take for all $i \geqslant 0, t_{i \alpha}=0$, and replace in the obtained $\{ W_{\beta}\}$, the loop parameter $z$ by $\frac{1}{z}$.
\end{remark}

\noindent
{\bf Conclusion}\\

\noindent
In this paper we considered three basic Lie algebras of $\mathfrak{t}$-loops, where $\mathfrak{t}$ is a maximal commutative complex Lie subalgebra of ${\rm sl}_{n}(\mathbb{C})$, but not necessarily a Cartan subalgebra. The generators of each of them are deformed in three different ways into ${\rm sl}_{n}$-loops depending of the commuting flows corresponding to these basic commutative Lie algebras. The first two deformations preserve commutativity, the third not, in general. We are interested in those deformations for which the evolution of the deformed generators is described by a specific set of Lax equations in each case. These three systems are shown to be compatible and can also be given in zero curvature form. This leads to the $({\rm sl}_{n}(\mathbb{C}), \mathfrak{t})$-hierarchy, its strict version and the combined $({\rm sl}_{n}(\mathbb{C}), \mathfrak{t})$-hierarchy, which can be seen as a merging of the first two. The combined $({\rm sl}_{n}(\mathbb{C}), \mathfrak{t})$-hierarchy is shown to have a linearization, which enables you to construct a wide class of solutions of this integrable hierarchy from a space of ${\rm sl}_{n}(\mathbb{C})$-loops.

\end{document}